\documentclass[12pt]{article}
\usepackage{graphicx}
\usepackage{amsmath}
\usepackage{amsthm}
\usepackage{amssymb}
\usepackage{amsfonts}
\usepackage{epsfig}

\def\Im {\mathop{\rm Im}\nolimits}
\def\arg {\mathop{\rm arg}\nolimits}
\def\Re {\mathop{\rm Re}\nolimits}
\def\Ai {{\rm Ai}}

\newtheorem{pro}{PROPOSITION}
\newtheorem{cor}{COROLLARY}
\newtheorem{lem}{LEMMA}
\newtheorem{thm}{THEOREM}
\numberwithin{equation}{section}

\title{{Gaussian unitary ensemble with boundary spectrum singularity and  $\sigma$-form of the Painlev\'{e} II equation }}

\author {Xiao-Bo Wu$^a$, Shuai-Xia Xu$^b$\footnote{Corresponding author (Shuai-Xia Xu).
 {\it{E-mail
address:}  xushx3@mail.sysu.edu.cn }}
  and  Yu-Qiu Zhao$^c$}
\date{
{\it{$^a$School of Mathematics and Computer Science, Shangrao Normal University, Shangrao 334001, China}}\\
 {\it{$^b$Institut Franco-Chinois de l'Energie Nucleaire, Sun Yat-sen University, GuangZhou
510275,  China}}\\
 {\it{$^c$Department of Mathematics, Sun Yat-sen University, GuangZhou
510275, China}}}

\begin{document}

\maketitle

\noindent \hrule width 6.27in\vskip .3cm

\noindent {\bf{Abstract }} We consider the Gaussian unitary ensemble perturbed by a Fisher-Hartwig singularity  simultaneously of both root type and jump type.
In the critical regime where the singularity approaches the soft edge, namely, the  edge of the support of the equilibrium measure for the Gaussian weight,
the asymptotics of the Hankel determinant and the recurrence coefficients, for the orthogonal polynomials associated with the
perturbed Gaussian weight, are obtained and expressed in terms of a family of smooth solutions to the Painlev\'{e} XXXIV equation and the  $\sigma$-form of the Painlev\'{e} II equation. In addition, we further obtain the double scaling limit of the distribution of the largest eigenvalue in a thinning procedure of the  conditioning Gaussian unitary ensemble, and the  double scaling limit of the correlation kernel for the critical  perturbed Gaussian unitary ensemble.
The asymptotic properties  of the  Painlev\'{e} XXXIV functions and the  $\sigma$-form of the Painlev\'{e} II equation are also studied.


\vspace{.2in} \noindent {\it {Keywords: Painlev\'{e} XXXIV equation, perturbed Gaussian unitary ensemble, Hankel determinant, orthogonal polynomials, Riemann-Hilbert approach}}

\noindent \hrule width 6.27in\vskip 1.3cm

\newpage

\section{Introduction and statement of results} \indent\setcounter{section} {1}

\noindent

We consider the perturbed Gaussian Unitary Ensemble (pGUE) defined by the following joint probability density function of the eigenvalues
\begin{equation}\label{pGUE}\rho_n(\lambda_1,...,\lambda_n)=\frac{1}{Z_{n,\alpha,\omega}}\prod_{i=1}^{n}w(\lambda_i)\prod_{1\leq i<j\leq n}^{n}|\lambda_i-\lambda_j|^2,\end{equation}
with $w(x)$ being the Gaussian weight perturbed by a Fisher-Hartwig singularity, that is,
\begin{equation}\label{weight}
w(x;\alpha,\mu,\omega)=e^{-x^2}|x-\mu|^{2\alpha}\left\{\begin{array}{cc}
                                                         1 & x\leq\mu\\
                                                         \omega & x>\mu,
                                                       \end{array}
\right.  \end{equation}
where $\alpha>-\frac 12$, the constant $ \omega\in \mathbb{C}\setminus(-\infty,0)$, and the partition function $Z_{n,\alpha,\omega}$ is a normalization constant.
It is readily seen that there is a  singularity at $x=\mu$, simultaneously of   root type and jump type.
  The case where  $\mu=\mu_n$ approaches  the soft edge  $\sqrt{2n}$ is of particular interest to us.

   It is well known that (\ref{pGUE}) admits a determinantal form
\begin{equation}\label{determiantal form} \rho_n(\lambda_1,...,\lambda_n)=\det[K_n(\lambda_i,\lambda_j)]_{1\leq i,j\leq n},\end{equation}
where
\begin{equation}\label{Weighted OP kernel}
K_n(x,y)=\sqrt{w(x)w(y)}\sum_{k=0}^{n-1}P_k(x)P_k(y),
\end{equation}
and $P_k(x)$ are  orthonormal polynomials   with
respect to the weight (\ref{weight}). For $\alpha>-\frac 12$ and $\omega\geq0$, the system of orthogonal  polynomials are well defined. We will prove later that the orthogonal  polynomials $P_n$ are also well defined for  $\alpha>-\frac 12$, $ \omega\in \mathbb{C}\setminus(-\infty,0)$ and $n$ large enough.
Let $P_n(x)=\gamma_n\pi_n(x)$, $\gamma_n$ being  the leading coefficient of the orthonormal polynomial, we have
the three-term recurrence relation
\begin{equation}\label{recurrence relations}
z\pi_n(z)=\pi_{n+1}(z)+a_n \pi_n(z)+b_n^2 \pi_{n-1}(z).
\end{equation}

Let $H_n(\mu;\alpha,\omega)$ be the Hankel determinant with respect to \eqref{weight}, that is,
\begin{eqnarray}
H_n(\mu;\alpha,\omega) &=& \det\left(\int_{-\infty}^{+\infty}x^{i+j}w(x)dx\right)_{i,j=0}^{n-1} \nonumber\\
 &=& \frac {1}{n!}\int_{-\infty}^{\infty}...\int_{-\infty}^{+\infty}\prod_{1\leq i<j\leq n}^{n}|\lambda_i-\lambda_j|^2\prod_{i=1}^{n}w(\lambda_i)d\lambda_i,\label{Hankel determiant}
\end{eqnarray}
then the normalization constant $Z_{n,\alpha,\omega}$ in \eqref{pGUE} is related to the Hankel determinant via
$$Z_{n,\alpha,\omega}=n!H_n(\mu;\alpha,\omega).$$

The pGUE arises naturally in  the statistics of eigenvalues of Gaussian unitary ensemble of random matrices.
 For $\alpha\in\mathbb{ N}$ and $\omega=1$, the pGUE can be interpreted as the probability density function of the classical Gaussian unitary ensemble under the condition  that $\mu$ is an eigenvalue with multiplicity $\alpha$ \cite{Forrester,ForresterWitte-2015}. Moreover, the distribution  of  the largest eigenvalue of this conditioning Gaussian unitary ensemble can be expressed as the ratio of two Hankel determinants with different parameters defined in \eqref{Hankel determiant}
 $$\mathrm{Pro}(\lambda_{max}\leq \mu\; |\; \lambda=\mu \mbox{ is an  eigenvalue  with  multiplicity } \alpha)=\frac {H_{n}(\mu;\alpha,0)}{H_{n}(\mu;\alpha,1)}.$$ Generally,  if we remove each eigenvalue of the  conditioning Gaussian unitary ensemble with probability $\omega\in [0,1]$, then the distribution  of the remaining largest eigenvalue is described as
\begin{equation}\label{gab pro}
\mathrm{Pro}\left (\lambda_{max}^{\mathrm{R}}\leq \mu \;|\; \lambda=\mu \mbox{ is an  eigenvalue  with  multiplicity } \alpha\right )=\frac {H_{n}(\mu;\alpha,\omega)}{H_{n}(\mu;\alpha,1)}.
\end{equation}
The study of thinning and conditioning random processes  can be found in      \cite{BogatskiyClaeysIts, bp,cc}.

 For the Gaussian unitary ensemble, there is the  celebrated  Tracy-Widom formula for the  distribution of the
 largest eigenvalue \cite{TracyWidom}
 \begin{equation}\label{Tracy-Widom formula}
 \lim_{n\to\infty}\mathrm{Pro}\left (\sqrt{2} n^{1/6} \left (\lambda_{max}-\sqrt{2n}\right )\leq s\right )=\exp\left(-\int_s^{\infty}(x-s)q^2(x)dx\right),
 \end{equation}
where $q(x)$ is the Hastings-Mcleod solution to the Painlev\'{e} II equation
\begin{equation}\label{Painleve II}
 q_{xx}=xq+2q^3,
 \end{equation}
determined by the boundary condition
\begin{equation}\label{HM-infinity}
 q(x)\sim \Ai (x),\quad  x\rightarrow +\infty.
  \end{equation}
It is noted that, if we take $\alpha=0$ and $\mu=\sqrt{2n}+\frac s{\sqrt{2}  n^{1/6}}$ in (\ref{pGUE}), then
$$\mathrm{Pro}\left (\lambda_{max}\leq \sqrt{2n}+\frac s{\sqrt{2}n^{1/6}}\right  )=\frac {H_{n}(\mu;0,0)}{H_{n}(\mu;0,1)}.$$
The Tracy-Widom distribution holds for general random matrix ensembles \cite{DeiftZhouU}, and is thus termed   a universality property.

In \cite{BogatskiyClaeysIts}, Bogatskiy, Claeys and Its studied  the thinned Gaussian unitary ensemble.  By removing  each eigenvalue of the Gaussian unitary ensemble with probability $1-k^2$, they  arrive at a thinning process. In our notations, the  distribution of the largest particle of the thinning process can by expressed as
$$\mathrm{Pro}\left (\lambda_{max}\leq \sqrt{2n}+\frac s{\sqrt{2}n^{1/6}}\right  )=\frac {H_{n}(\mu;0,1-k^2)}{H_{n}(\mu;0,1)}.$$
The  large $n$ asymptotic approximation of the distribution was proved in \cite{BogatskiyClaeysIts},
 \begin{equation}\label{g-Tracy-Widom formula}
 \mathrm{Pro}\left (\lambda_{max}\leq \sqrt{2n}+\frac s{\sqrt{2}n^{1/6}} \right )=\exp\left(-\int_s^{\infty}(x-s)q^2(x)dx\right)(1+o(1)),
 \end{equation}
where  $q(x)$ is the Ablowitz-Segur solution to the Painlev\'{e} II equation determined by the boundary condition
\begin{equation}\label{Ablowitz-Segur-infinity }
 q(x)\sim k \Ai (x),\quad  x\rightarrow +\infty.
  \end{equation}
The Tracy-Widom type formula (\ref{g-Tracy-Widom formula}) was obtained in \cite{BogatskiyClaeysIts} by studying the asmptotics of the orthogonal polynomials with respect to (\ref{weight}) via the Riemann-Hilbert approach, specifying  $\alpha=0$. This system of orthogonal polynomials  has  also been  considered by Xu and Zhao \cite{XuZhao}, in which the asymptotics of recurrence coefficients were obtained, also by applying  the Riemann-Hilbert approach.

In \cite{ItsKuijlaarsOstensson2008}, Its, Kuijlaars and \"{O}stensson studied pGUE (\ref{pGUE}) in the case $\omega=0$.
When the  algebra singularity $\mu$ is close to the soft edge, or, more precisely, when   $\mu=\sqrt{2n}+\frac s{\sqrt{2}n^{1/6}}$ with mild $s$, the asymptotics of the correlation kernel (\ref{Weighted OP kernel}) was found and characterised by  the solution to a Lax pair,  related to a certain Painlev\'{e} XXXIV transcendent. The authors of \cite{ItsKuijlaarsOstensson2008} showed that the Painlev\'{e} XXXIV kernel is valid   for quite general ensembles with root-typed  singularity $|z-\mu|^{2\alpha}$, $\mu$ being close to the soft edge.

Earlier in \cite{ForresterWitte}, Forrester and Witte considered   pGUE (\ref{pGUE}) with the parameter $2\alpha\in \mathbb{N}$ and $\omega=0,\; 1$. By applying the Okamoto $\tau$-function theory, they found  that the asymptotics of the quantities $H_n(\mu;\alpha,\omega)$ in  (\ref{Hankel determiant}), with $\mu=\sqrt{2n}+\frac s{\sqrt{2}n^{1/6}}$, can be expressed in terms of solutions of the following Jimbo-Miwa-Okamoto $\sigma$-form of the Painlev\'{e} II equation:
\begin{equation}\label{sigma form}
(\sigma'')^2+4(\sigma')^3-4s(\sigma')^2+4\sigma'\sigma- (2\alpha)^2=0,
 \end{equation}
 with the boundary condition
 \begin{equation}\label{sigma form-boundary-1}
\sigma(s)\sim \frac {s^2}{4}+\frac {16\alpha^2-1} {8}s^{-1},~~ s\to -\infty
 \end{equation}
 for $\omega=0$,
 and
 \begin{equation}\label{sigma form-boundary-2}
\sigma(s)\sim -2\alpha s^{1/2}-\frac {\alpha^2}{s} ,~~ s\to+ \infty
 \end{equation}
 for $\omega=1$. In \cite{ForresterWitte}, the authors also noted that for $\alpha=0$,
\begin{equation}\label{sigma form-painleve}
\sigma'=-q^2,
 \end{equation}
and they further obtained (\ref{g-Tracy-Widom formula})-(\ref{Ablowitz-Segur-infinity }).
In \cite{ForresterWitte-2002} and \cite{XuZhao-2015,ZengXuZhao},  asymptotic formulas  have been derived  of the Hankel determinants associated   with   the Jacobi weight perturbed by  a Fisher-Hartwig
singularity  close to the hard edge $x=1$,    involving  the Jimbo-Miwa-Okamoto $\sigma$-form of the Painlev\'e III equation.
 The Painlev\'{e} equations   play an important role in the asymptotic study  of the Hankel determinants; see \cite{DeiftItsKrasovky,ForresterWitte-2015,XDZ-1,XDZ-2}.

In this paper, we consider   pGUE \eqref{pGUE} with Fisher-Hartwig singularity of both root type and jump type simultaneously at $\lambda=\mu_n$,
which approaches  the soft edge $\sqrt{2n}$ at a certain speed such that $\sqrt{2}n^{1/6} (\mu_n-\sqrt{2n}\,  )\to s$ with finite $s$, as $n\to\infty$. We derive the asymptotics of the
Hankel determinant, the recurrence coefficients and the correlation kernel in the double scaling limit. The asymptotic results  are expressed in terms of a family of
solutions to the  Jimbo-Miwa-Okamoto $\sigma$-form of the Painlev\'e II equation and the Painlev\'e XXXIV equation. The asymptotic properties  of this family of solutions  are also investigated.

\subsection{Statement of main results}

\subsubsection{Asymptotics of  solutions to the Painlev\'e XXXIV equation and the $\sigma$-form of the Painlev\'{e} II equation }
Our first result is on the asymptotics of a family of  solutions to the Painlev\'e XXXIV equation  and the Jimbo-Miwa-Okamoto $\sigma$-form of the Painlev\'e II equation. The asymptotics at positive infinity are derived by using known results on the second Painlev\'{e} transcendents in \cite{FokasItsKapaevNovokshenov}.  These solutions are  used to describe the asymptotics of the Hankel determinants and the recurrence coefficients.
\begin{thm}\label{Asymptotic of sigma form}
Let $\alpha>-\frac 12$ and $\omega\in \mathbb{C}\setminus (-\infty,0)$,
then for $s\in \mathbb{R}$, there is an analytic solution to the Jimbo-Miwa-Okamoto $\sigma$-form of the Painlev\'{e} II equation (\ref{sigma form}),  uniquely determined by  the boundary condition as $s\to \infty$ and $\arg s\in (-\frac {\pi}{3}, \frac {\pi }{3}]$, namely
\begin{equation}\label{thm: sigma asy positive}
 \begin{split}
  \sigma(s)=&-2\alpha s^{1/2} \left[\sum_{k=0}^m a_ks^{-\frac {3}{2}k }+O\left (s^{-\frac {3(m+1)}{2}}\right )\right]\\
  &+\left (e^{2\pi i\alpha}-\omega\right )\frac {\Gamma(2\alpha+1)}{2^{3+6\alpha}\pi }s^{-(3\alpha+1)}e^{-\frac 43 s^{3/2}} \left [1+O\left (s^{-\frac {1}{4}}\right ) \right ],
  \end{split}
\end{equation}
where $m$ is a positive integer,   $a_0=1$, and  $a_1=\frac {\alpha}{2}$.
Furthermore, the asymptotic behavior of the solution  at negative infinity
is
\begin{equation}
\sigma(s)=\frac {s^2}{4}+\frac {16\alpha^2-1} {8}s^{-1}+O(s^{-5/2})~~\mbox{as}~s\to -\infty
\end{equation}
for $\omega=0$,  and
\begin{align}\label{thm: sigma asy negative infinity} \sigma(s)=&2\beta i (-s)^{1/2}
+\frac{i}{4s}\left[2i(\beta^2-\alpha^2)+\frac{\Gamma(1+\alpha-\beta)}{\Gamma(\alpha+\beta)}
e^{i\theta(s;\alpha,\beta)}-\frac{\Gamma(1+\alpha+\beta)}{\Gamma(\alpha-\beta)}
e^{-i\theta(s;\alpha,\beta)}\right]\nonumber\\
&+O(s^{3|\Re \beta|-\frac 52})~~\mbox{as}~s\to -\infty
\end{align}
for $ \omega=e^{-2\pi i\beta}$ with $|\Re \beta|<1/2$,
where $\theta(s;\alpha,\beta)=\frac 43 |s|^{3/2}-\alpha \pi-6i \beta \ln 2-3i\beta\ln|s|$.
\end{thm}

\begin{thm}\label{Asymptotic of Painleve}Let $\alpha>-\frac 12$ and $\omega\in \mathbb{C}\setminus (-\infty,0)$,
then for $s\in \mathbb{R}$, there is an analytic solution to the Painlev\'{e} XXXIV equation
\begin{equation}\label{pXXXV}
 {u}_{ss}=\frac{{u}_s^2}{2{u}}+4{u}^2+2s{u}
 -\frac{(2\alpha)^2}{2{u}},
  \end{equation}
 uniquely determined by  the boundary condition as $s\to \infty$ and $\arg s\in (-\frac {\pi}{3}, \frac {\pi }{3}]$, namely
\begin{equation}\label{thm: painleve  positive infinity}
\begin{split}
  u(s)=&\frac {\alpha}{\sqrt{s}}\left [\sum_{k=0}^m c_ks^{-\frac {3}{2}k}+O\left (s^{-\frac {3(m+1)}{2}}\right )\right]\\
  &+\left (e^{2\pi i\alpha}-\omega\right )\frac {\Gamma(2\alpha+1)}{2^{2+6\alpha}\pi }s^{-(3\alpha+\frac 12)}e^{-\frac 43 s^{3/2}}\left [1+O\left (s^{-\frac {1}{4}}\right )\right ],
\end{split}
\end{equation}
where $m$ is a positive integer, $c_0=1$ and $c_1=-\alpha$.
Furthermore,   the asymptotic behavior of the solution  at negative infinity
is
\begin{equation}
u(s)=-\frac {s}{2}+\frac {16\alpha^2-1} {8}s^{-2}+O(s^{-7/2}) ~~\mbox{as}~s\to -\infty
\end{equation}
for $\omega=0$, and
\begin{equation} \label{thm: painleve  negative infinity} u(s)=\frac{1}{\sqrt{-s}}\left[i\beta+\frac 12 \frac{\Gamma(1+\alpha-\beta)}{\Gamma(\alpha+\beta)}
e^{i\theta(s;\alpha,\beta)}+\frac 12 \frac{\Gamma(1+\alpha+\beta)}{\Gamma(\alpha-\beta)}
e^{-i\theta(s;\alpha,\beta)}\right]+O(s^{3|\Re \beta|-2})
\end{equation}as $s\to -\infty$
for $ \omega=e^{-2\pi i\beta}$ with $|\Re \beta|<1/2$,
where $\theta(s;\alpha,\beta)=\frac 43 |s|^{3/2}-\alpha \pi-6i \beta \ln 2-3i\beta\ln|s|$.
\end{thm}

\subsubsection{Asymptotics of the Hankel determinant and recurrence coefficients}

Our second result is on the asymptotics of the Hankel determinant,  characterized by the   Jimbo-Miwa-Okamoto $\sigma$-form of the Painlev\'{e} II equation.
\begin{thm}\label{Asymptotic of Hankel determinant}
Let $\alpha>-\frac 12$, $\omega\in \mathbb{C}\setminus (-\infty,0)$ and $\mu_n=\sqrt{2n}+\frac s{\sqrt{2}n^{1/6}}$, we have the following asymptotic formula  of the Hankel determinant
\begin{equation}\label{Hankel-Asy}
\frac {d}{d\mu_n} \ln H_n(\mu_n;\alpha,\omega)=
\sqrt{2n}\left(2\alpha +\frac{\sigma(s)}{n^{1/3}} +\frac{\alpha(u(s)+s) }{n^{2/3}}+O\left (\frac 1 n\right )\right),
\end{equation}
where  $\sigma(s)$ and  $u(s)$  are the analytic solutions to  the   Jimbo-Miwa-Okamoto $\sigma$-form of the Painlev\'{e} II equation (\ref{sigma form})  and the Painlev\'e XXXIV equation,  determined   respectively  by
 the boundary condition  \eqref{thm: sigma asy positive} and \eqref{thm: painleve  positive infinity}.
\end{thm}

The above theorem  can be applied directly to obtain  the double scaling limit of the distribution of the largest eigenvalue in the conditioning and thinning procedure in GUE; cf. \eqref{gab pro}. The limit distribution  is connected to the  perturbed GUE \eqref{pGUE}, and can also  be expressed in terms  of the Hankel determinant in \eqref{Hankel determiant}.
\begin{cor}\label{thm: TW for pGUE}
Let $\alpha>-\frac 12$, $\omega\in \mathbb{C}\setminus (-\infty,0)$ and $\mu_n=\sqrt{2n}+\frac s{\sqrt{2}n^{1/6}}$ with finite $s$, we have the following double scaling limit for the distribution of the largest eigenvalue remaining  in the thinning procedure of the conditioning  GUE  defined in \eqref{gab pro}:
\begin{equation}\label{TW for pGUE}
\begin{split}
   & \lim_{n\to+\infty} \mathrm{Pro} \left (\lambda_{max}^{\mathrm{R}}\leq \mu_n \;|\; \lambda=\mu_n \mbox{ is an  eigenvalue  with  multiplicity } \alpha\right ) \\
    &  =\lim_{n\to+\infty}\frac {H_{n}(\mu_n;\alpha,\omega)}{H_{n}(\mu_n;\alpha,1)}\\
    &=\exp\left(\int_{s}^{+\infty}\sigma(t;\alpha,1)-\sigma(t;\alpha,\omega)dt\right),
\end{split}
\end{equation}
 where  $\sigma(s;\alpha,\omega)$ are a family of  analytic solutions to  the   Jimbo-Miwa-Okamoto $\sigma$-form of the Painlev\'{e} II equation (\ref{sigma form}),  determined    by  the boundary condition  \eqref{thm: sigma asy positive}, and depending on the parameter $\alpha$ and $\omega$.
\end{cor}
For $\alpha=0$, the Jimbo-Miwa-Okamoto $\sigma$-form  is related to the  Painlev\'{e} II equation by $\sigma'(s)=-q(s)^2$; see \eqref{sigma form-painleve}, we can accordingly derive from \eqref{TW for pGUE} and  \eqref{thm: sigma asy positive} the Tracy-Widom distribution of the largest eigenvalue for GUE \eqref{Tracy-Widom formula}-\eqref{HM-infinity} and for   the thinned  GUE \eqref{g-Tracy-Widom formula}-\eqref{Ablowitz-Segur-infinity },  which have previously been obtained in \cite{TracyWidom} and \cite{BogatskiyClaeysIts}, respectively.

We also have the following asymptotic approximations  of the recurrence coefficients in \eqref{recurrence relations},  in terms of  the Painlev\'e XXXIV transcendent.

\begin{thm}\label{Asymptotic of coefficients}
Let $\alpha>-\frac 12$, $\omega\in \mathbb{C}\setminus (-\infty,0)$ and $\mu_n=\sqrt{2n}+\frac s{\sqrt{2}n^{1/6}}$, we have
\begin{equation} \label{an-asymptotic}
  a_n=-\frac {u(s)}{2^{1/2}}n^{-1/6}+O(n^{-1/2})
\end{equation}
and \begin{equation} \label{bn-asymptotic}
   b_n=\frac{1}{\sqrt{2}}n^{1/2}-\frac {u(s)}{2^{3/2}}n^{-1/6}+O(n^{-1/2})
\end{equation}as $n\to\infty$,
where  $u(s)$  is the Painlev\'e XXXIV transcendent \eqref{pXXXV}, analytic for $s\in \mathbb{R}$ and subject to the boundary condition  \eqref{thm: painleve  positive infinity}.
\end{thm}

\subsubsection{Asymptotics of the correlation kernel}
Let $\left(
                                \begin{array}{c}
                                  \psi_1(x;s) \\
                                  \psi_2(x;s)\\
                                \end{array}
                              \right)$ be the  unique solution to the  linear differential equation
\begin{align}
\frac {\partial}{\partial x}\left(
                                \begin{array}{c}
                                  \psi_1(x;s) \\
                                  \psi_2(x;s)\\
                                \end{array}
                              \right)=&\left(
      \begin{array}{cc}
        1 & 0 \\
       -i(\sigma(s)-\frac {s^2}{4}) & 1 \\
      \end{array}
    \right)
  \left(
      \begin{array}{cc}
        \frac {u'}{2x} & i-i\frac {u}{x} \\
        -i x-i(u+s)-i\frac {(u')^2-(2\alpha)^2}{4u x} & -\frac {u'}{2x} \\
      \end{array}
    \right)\\\nonumber
    &\times \left(
      \begin{array}{cc}
        1 & 0 \\
       i(\sigma(s)-\frac {s^2}{4}) & 1 \\
      \end{array}
    \right)
\left(
                                \begin{array}{c}
                                  \psi_1(x;s) \\
                                  \psi_2(x;s)\\
                                \end{array}
                              \right),
\end{align}
subject to the boundary conditions
\begin{equation}\label{generalized airy kernel-infty}\left(
                                \begin{array}{c}
                                  \psi_1(x;s) \\
                                  \psi_2(x;s)\\
                                \end{array}
                              \right)
=\frac{\omega^{\frac12}}{\sqrt{2}}e^{-\left (\frac{2}{3}x^{\frac{3}{2}}+s\sqrt{x}\right )}\left(
   \begin{array}{c}
    x^{-\frac{1}{4}}+O(x^{-\frac {3}{4}})\\
     ix^{\frac{1}{4}}+O(x^{-\frac {1}{4}})\\
   \end{array}
 \right)
 ~~\mbox{as}~x\rightarrow+\infty;\end{equation}
 and
\begin{equation}\label{generalized airy kernel- negative infinity}\left(
                                \begin{array}{c}
                                  \psi_1(x;s) \\
                                  \psi_2(x;s)\\
                                \end{array}
                              \right)
=\sqrt{2}\left(
   \begin{array}{c}
    |x|^{-\frac{1}{4}}\cos(\frac{2}{3}|x|^{\frac{3}{2}}-t\sqrt{|x|}-\alpha\pi -\pi/4)+O(|x|^{-\frac 34})\\
     -i |x|^{\frac{1}{4}}\sin(\frac{2}{3}|x|^{\frac{3}{2}}-t\sqrt{|x|}-\alpha\pi -\pi/4)+O(|x|^{-\frac 14})\\
   \end{array}
 \right)
 ~~\mbox{as}~x\rightarrow-\infty.\end{equation}
The $\Psi$-kernel is defined as
\begin{equation}\label{generalized airy kernel-integrable form}K_{\Psi}(x,y;s)=\frac{\psi_2(x;s)\psi_1(y;s)-\psi_1(x;s)\psi_2(y;s)}{2\pi i(x-y)}.\end{equation}
Then, our result on the  correlation kernel can be stated as follows.
\begin{thm}\label{Asymptotic of kernel}
Let $\alpha>-\frac 12$, $\omega\in \mathbb{C}\setminus (-\infty,0)$ and $\mu_n=\sqrt{2n}+\frac s{\sqrt{2}n^{1/6}}$, we have the limit  of the correlation kernel
\begin{equation}\label{kernel-asymptotic}\lim_{n\to\infty}\frac 1{\sqrt{2}n^{1/6}}K_n(\mu_n+\frac {v_1}{\sqrt{2}n^{1/6}},\mu_n+\frac {v_2}{\sqrt{2}n^{1/6}})=K_{\Psi}(v_1,v_2)
\end{equation}
for $v_1,v_2\in \mathbb{R}\setminus \{0\}$.
\end{thm}

The rest of the  paper is organized as follows. In Section \ref{sec:2}, we introduce the Riemann-Hilbert problem (RH problem, for short) for the  Painlev\'{e} XXXIV equation. The existence of a unique solution to the RH problem is proved and thus the related Painlev\'{e} XXXIV tanscendents and $\sigma$-form of the Painlev\'{e} II equation are pole free on the real axis. The asymptotics of the Painlev\'e XXXIV functions  are also discussed in this section.  In Section \ref{sec:3}, we state the RH problem for the orthogonal polynomials associated with the weight \eqref{pGUE}. We show that the Hankel determinant of finite size satisfies the  $\sigma$-form of the Painlev\'e IV equation. A differential identity for the Hankel determinant is also derived.  The large-$n$ asymptotic analysis of the RH problem for the orthogonal polynomials is then carried out by applying  the Deift-Zhou method. In Section \ref{sec:4}, using the asymptotic results of the RH problem for the orthogonal polynomials obtained in the previous section, we obtain  the asymptotics of the Hankel determinant   and the recurrence coefficients, expressed in terms of the $\sigma$-form of the Painlev\'{e} II equation and the Painlev\'e XXXIV tanscendents. An asymptotic formula of the correlation kernel are also proved in this section. These furnish proofs of Theorems \ref{Asymptotic of Hankel determinant}, \ref{Asymptotic of coefficients} and \ref{Asymptotic of kernel}. The proofs of the rest theorems, namely Theorems \ref{Asymptotic of sigma form} and \ref{Asymptotic of Painleve},  are left to the Appendix. In the Appendix, the asymptotics of the Painlev\'{e} XXXIV tanscendents and $\sigma$-form of the Painlev\'{e} II equation as $s\to +\infty$ are derived  by analyzing  the RH problem  related to the Painlev\'e XXXIV equation.


 \section{Riemann-Hilbert problem for the  Painlev\'{e} XXXIV equation}
 \indent\setcounter{section} {2}
\setcounter{equation} {0} \label{sec:2}
The   matrix-valued RH  problem for  the Painlev\'e XXXIV equation is as follows; see \cite{FokasItsKapaevNovokshenov, ItsKuijlaarsOstensson2008,ItsKuijlaarsOstensson2009}.

\begin{figure}[ht]
\begin{center}
\resizebox*{6cm}{!}{\includegraphics{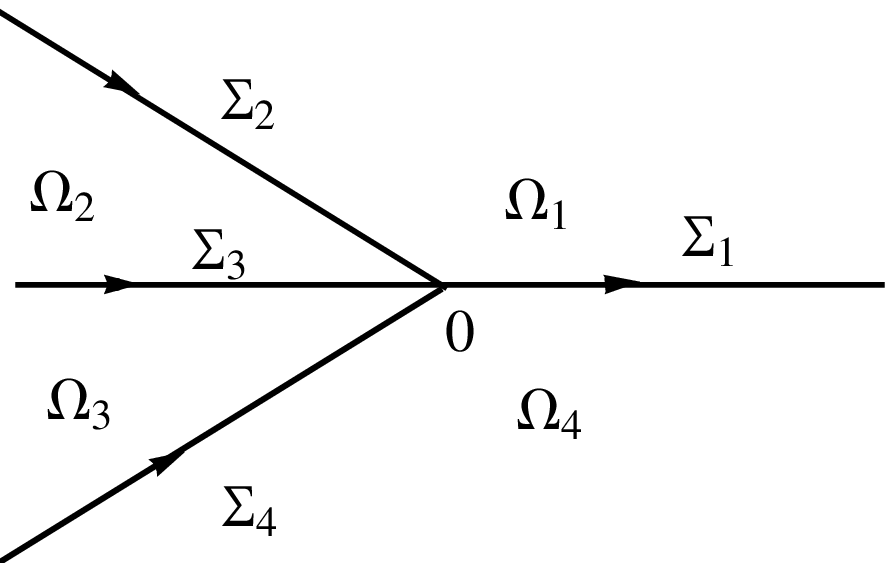}}
\caption{Regions and contours for $\Psi$}
\label{fig-P34} \end{center}
\end{figure}

\begin{description}
  \item(a)~~  $\Psi(\zeta;s)$ ($\Psi(\zeta)$, for short) is analytic in
  $\mathbb{C}\backslash \bigcup_{j=1}^4\Sigma_j$; see Fig.\;\ref{fig-P34} for the jump contours;

  \item(b)~~  $\Psi(\zeta)$  satisfies the jump condition
  \begin{gather}\label{Psi-jump}
  \Psi_+ (\zeta)=\Psi_- (\zeta)
  \left\{ \begin{array}{ll}
            \left(
                               \begin{array}{cc}
                                 1 & \omega \\
                                 0 & 1 \\
                                 \end{array}
                             \right), & \zeta \in {\Sigma}_1, \\[.4cm]
           \left(
                               \begin{array}{cc}
                                 1 & 0 \\
                                 e^{2\alpha\pi i} & 1 \\
                                 \end{array}
                             \right), & \zeta \in {\Sigma}_2,\\[.4cm]
           \left(
                               \begin{array}{cc}
                                 0 & 1 \\
                                 -1 & 0 \\
                                 \end{array}
                             \right),&
                                           \zeta \in {\Sigma}_3,\\  [.4cm]
            \left(
                               \begin{array}{cc}
                                 1 & 0 \\
                                 e^{-2\alpha\pi i} & 1 \\
                                 \end{array}
                             \right),  & \zeta \in \Sigma_4;
          \end{array}
    \right .
  \end{gather}
\item(c)~~     As $\zeta\rightarrow \infty$
\begin{equation}\label{psi-infinity}
 \begin{array}{rl}
   \Psi(\zeta)& =
\zeta^{-\frac{1}{4}\sigma_3}M_0
\left [I+\frac{\hat{\sigma}(s)(\sigma_3-i\sigma_1) }{2\sqrt{\zeta}}-
  \frac{(\hat{\sigma}'(s) +\hat{\sigma}^2(s))\sigma_2}{2\zeta}+O\left (\zeta^{-\frac{3}{2}}\right )\right ]
e^{-\vartheta \sigma_3} \\[.3cm]
  &=     \left [I+\frac 1{2\zeta}
     \left(\begin{array}{cc}\hat{ \sigma}'(s) +\hat{\sigma}^2(s) & -2i\hat{\sigma}(s)\\ * & -(\hat{\sigma}'(s) +\hat{\sigma}^2(s))\end{array} \right)+O\left (  \zeta^{-2}\right )\right]
   \zeta^{-\frac{1}{4}\sigma_3}M_0  e^{-\vartheta \sigma_3} ,
 \end{array}
\end{equation}
  where $M_0= \frac{I+i\sigma_1}{\sqrt{2}}$,     $\vartheta=\vartheta(\zeta,s)=\frac{2}{3}\zeta^{3/2}+s\zeta^{1/2}$,
  and $\sigma_j$ are the Pauli matrices, namely,
  \begin{equation*}
\sigma_1=\left(
                   \begin{array}{cc}
                     0 &1 \\
                    1 & 0 \\
                   \end{array}
                 \right),   ~~\sigma_2=\left(
                   \begin{array}{cc}
                     0 & -i \\
                    i & 0 \\
                   \end{array}
                 \right),~~\mbox{and}~\sigma_3=\left(
                   \begin{array}{cc}
                     1 & 0 \\
                    0 & -1 \\
                   \end{array}
                 \right);
 \end{equation*}
\item(d)~~ As $\zeta\rightarrow 0$, $\zeta \in \Omega_j$ for $j=1,2,3,4$,
\begin{equation}\label{Psi0-1}
 \Psi(\zeta)=\Psi_0(\zeta)\zeta^{\alpha \sigma_3} E_j~~\mbox{if}~2\alpha\notin   \mathbb{N}_0,
 \end{equation}or
 \begin{equation}\label{Psi0-2}
\Psi(\zeta)= \Psi_0(\zeta)\zeta^{\alpha \sigma_3}\left (I+\frac k 2(\ln \zeta +\pi)\sigma_+\right ) E_j~~\mbox{if}~2\alpha\in   \mathbb{N}_0,
\end{equation}
where  $\sigma_+=\left(
                   \begin{array}{cc}
                     0 & 1 \\
                     0 & 0 \\
                   \end{array}
                 \right)$, $\Psi_0(\zeta)$ is analytic at $\zeta=0$, and
$ E_j$ are certain constant matrices; see  \cite[(3.21)]{ItsKuijlaarsOstensson2008}.
\end{description}

\begin{pro}
For $\alpha>-\frac 12$ and $\omega\in \mathbb{C}\setminus (-\infty,0)$, there exits a unique solution to the above Riemann-Hilbert problem for $\Psi(\zeta;s)$.
\end{pro}
\begin{proof}
The  existence of the solution   $\Psi(\zeta)$ is obtained by proving a vanishing lemma; see
  \cite{ItsKuijlaarsOstensson2009} and \cite{XuZhao} for a proof of the vanishing lemma.
\end{proof}

Now we observe that the entries of $\Psi(\zeta)$  are closely related with equations \eqref{sigma form} and  \eqref{pXXXV} as follows:

\begin{pro}\label{Sigma form and psi}  Let
\begin{equation}\label{Sigma form and psi-1}
\sigma(s)=i \lim_{\zeta\rightarrow \infty}\zeta \left(\Psi(\zeta;s)e^{\vartheta\sigma_3}\frac{I-i\sigma_1}{\sqrt{2}}\zeta^{\frac{1}{4}\sigma_3} \right)_{12}+\frac {s^2}{4},
\end{equation}
  then $\sigma(s)$ satisfies the Jimbo-Miwa-Okamoto $\sigma$-form of the Painlev\'{e} II equation \eqref{sigma form}. Also,
\begin{equation*}
u(s)=-
 \sigma'(s)
 \end{equation*}
 satisfies the Painlev\'e XXXIV equation \eqref{pXXXV}.
 Moreover, for $\alpha>-\frac 12$ and $\omega\in \mathbb{C}\setminus (-\infty,0)$, the scalar functions $\sigma(s)$ and $u(s)$ are analytic for $s\in \mathbb{R}$.
\end{pro}

Let
$\widehat{\Psi}(\zeta;s)=\left(
                  \begin{array}{cc}
                    1 & 0\\
                    i(\sigma(s)-\frac {s^2}{4}) & 1 \\
                  \end{array}
                \right)\Psi(\zeta;s)
$, then $\widehat{\Psi}(\zeta;s)$  satisfies the following  Lax pair system of differential equations
\begin{align*}
&\frac {\partial}{\partial\zeta}\widehat{\Psi}(\zeta;s)=\left(
      \begin{array}{cc}
        \frac {u'}{2\zeta} & i-i\frac {u}{\zeta} \\
        -i \zeta-i(u+s)-i\frac {(u')^2-(2\alpha)^2}{4u \zeta} & -\frac {u'}{2\zeta} \\
      \end{array}
    \right)\widehat{\Psi}(\zeta;s), \nonumber \\
&\frac {\partial}{\partial s}\widehat{\Psi}(\zeta;s)=\left(
      \begin{array}{cc}
        0 & i\\
        -i\zeta-2i(u+\frac s2) & 0 \\
      \end{array}
    \right)\widehat{\Psi}(\zeta;s);
\end{align*}
see \cite[Lemma 3.2]{ItsKuijlaarsOstensson2008}.

It is well known that  the Painlev\'e XXXIV equation \eqref{pXXXV}  is related to the   Painlev\'{e} II equation
\begin{equation}\label{PII}
 q''=sq+2q^3-(2\alpha+\frac {1}{2})\end{equation}
 by the formulas
 \begin{equation}\label{PII-xxxiv}
 u(s)=2^{\frac {1}{3}} U(-2^{-\frac {1}{3}}s),\quad U(s)=q'+q^2+\frac {s}{2};\end{equation}
 see \cite[(5.0.2)(5.0.55)]{FokasItsKapaevNovokshenov}.
Moreover,  the solution $u(s)$ associated with the model RH problem, as stated in Proposition \ref{Sigma form and psi}, is related to a family of
special solutions to the second Painlev\'{e} equation \eqref{PII} corresponding to the Stokes multipliers $\{s_1= -e^{-2\pi i\alpha}, s_2=\omega, s_3=-e^{2\pi i \alpha}\}$; cf. \cite[ 11.3]{FokasItsKapaevNovokshenov}, see also \cite{ItsKuijlaarsOstensson2008,XuZhao} for the relation between the model RH problem for $\Psi$ and the RH problem for the
second Painlev\'{e} equation. Thus we may extract  the asymptotics  of $u(s)$,     knowing  that  of $q(s)$.

It is also known that the solutions $q(s)$ corresponding to the Stokes multipliers $\{s_1= -e^{-2\pi i\alpha}, s_2=\omega, s_3=-e^{2\pi i \alpha}\}$
is uniquely determined by the following asymptotic behavior
\begin{equation}\label{PII-asy}
 q(s)=\sqrt{\frac {-s}{2}}\sum_{k=0}^nb_k(-s)^{\frac {3}{2}k}+O\left (s^{-\frac {3}{2}n-1}\right )+c_{+}(-s)^{-3\alpha-1}e^{-\frac {2\sqrt{2}}{3}(-s)^{\frac {3}{2}}}\left (1+O\left (s^{-\frac {1}{4}}\right )\right ),\end{equation}
 as $\arg (-s)\in(-\frac {\pi}{3}, \frac {\pi}{3}]$ and $s\to \infty$, \cite[(11.5.56)]{FokasItsKapaevNovokshenov}.
 The  first few   coefficients  are
 $$b_0=1,\quad  b_1=\frac {\alpha}{\sqrt{2}},\quad  b_2=-\frac {1+6\alpha^2}{8},$$
 and
 $$c_+=\frac {e^{2\pi i\alpha}-\omega}{\pi} 2^{-5\alpha-3}\Gamma(2\alpha+1);$$
 see  \cite[(11.5.58) and (11.5.68)]{FokasItsKapaevNovokshenov}.
 A combination of  the asymptotics of $q(s)$ and the relation \eqref{PII-xxxiv} will then gives   the asymptotic approximations of $u(s)$ and $\sigma(s)$ as
 $s\to\infty$ and $\arg s\in(-\frac {\pi}{3}, \frac {\pi}{3}]$, as stated  in \eqref{thm: sigma asy positive} and \eqref{thm: painleve  positive infinity}.

We also mention that, for $\alpha=0$ and $\omega=1$, the solution $q(x)$ determined by \eqref{PII-asy} or the Stokes
 multipliers $\{s_1= -1, s_2=1, s_3=-1\}$ is the classical solution
 $$q(x)=-2^{-1/3}\frac {\Ai'(-2^{-1/3}x)}{\Ai(-2^{-1/3}x)};$$
 cf. \cite[(11.4.15)]{FokasItsKapaevNovokshenov}.
Thus  by \eqref{PII-xxxiv}  the Painlev\'{e} XXXIV function corresponding to the parameter $\alpha=0$ and $\omega=1$ is trivial, that is, $u(x)=0$; see also \cite[Lemma 3]{XuZhao}.

 The asymptotics of $q(s)$ as $s\to +\infty$ can be found  in, e.g. \cite[Chapter 10]{FokasItsKapaevNovokshenov}. However, the leading asymptotics of
 $q(s)$  is canceled out  when we substitute  it into  \eqref{PII-xxxiv}. Therefore,  the leading term alone  is not enough to derive the asymptotics of $u(s)$; see also the discussion in \cite{ItsKuijlaarsOstensson2008,ItsKuijlaarsOstensson2009}. We will derive the asymptotics of $u(s)$ and $\sigma(s)$ as
 $s\to-\infty$ by using  the  RH problem for $\Psi$. Details will be  given in the Appendix.

\section{Nonlinear steepest descend  analysis of orthogonal polynomials}
\indent\setcounter{section} {3} \label{sec:3}
In this section, we start by presenting  the RH problem $Y(z;\mu,n)$ for the orthogonal polynomials associated with the weight function \eqref{weight}.
Then we show that the logarithmic derivative of the Hankel determinants can be expressed in terms of the solutions to this RH problem.
For $n$ fixed, by relating the solution $Y(z;\mu,n)$ to the Lax pair for the Painlev\'e IV equation, we show that the logarithmic derivative of the Hankel determinant satisfies the Jimbo-Miwa-Okamoto $\sigma$-form of  the Painlev\'e IV equation. Then we perform a nonlinear steepest descent analysis of
the RH problem for $Y$ via the Deift-Zhou method \cite{DeiftZhouA,DeiftZhouU,DeiftZhouS} .

\subsection{Riemann-Hilbert problem for orthogonal polynomials }
The RH problem for the orthogonal polynomials with respect to \eqref{weight} is as follows.
\begin{description}
  \item(Y1)~~  $Y(z)$ is analytic in
  $\mathbb{C}\backslash \mathbb{R}$;

  \item(Y2)~~  $Y(z)$  satisfies the jump condition
  \begin{equation}\label{}
  Y_+(x)=Y_-(x) \left(
                               \begin{array}{cc}
                                 1 & w(x) \\
                                 0 & 1 \\
                                 \end{array}
                             \right),
\qquad x\in \mathbb{R},\end{equation} where $w(x)$ is defined in \eqref{weight};

  \item(Y3)~~  The behavior of $Y(z)$ at infinity is
  \begin{equation}\label{Y-infinity}Y(z)=\left (I+\frac {Y_{1}}{z}+\frac {Y_{2}}{z^2}+O\left (\frac 1 {z^3}\right )\right )\left(
                               \begin{array}{cc}
                                 z^n & 0 \\
                                 0 & z^{-n} \\
                               \end{array}
                             \right),\quad \quad z\rightarrow
                             \infty ;\end{equation}
\item(Y4)~~ The behavior of $Y(z)$ at $\mu$ is
\begin{equation}\label{}Y(z)= \left\{\begin{array}{ll}
O\left(
                               \begin{array}{cc}
                                 1 & |z-\mu|^{2\alpha} \\
                                 1 & |z-\mu|^{2\alpha} \\
                               \end{array}\right), & z\rightarrow\mu,\quad \alpha <0,\\[.4cm]
O\left(
                               \begin{array}{cc}
                                 1 & \log|z-\mu| \\
                                 1 & \log|z-\mu| \\
                               \end{array}\right), & z\rightarrow\mu,\quad \alpha =0,\\[.4cm]
O\left(
                               \begin{array}{cc}
                                 1 & 1 \\
                                 1 & 1 \\
                               \end{array}\right), & z\rightarrow\mu,\quad \alpha >0.
 \end{array}\right.\end{equation}
\end{description}

For $\alpha> -\frac {1}{2}$ and $\omega\geq0$, the system of orthogonal polynomials with respect to \eqref{weight} are well defined. It follows from the Sochocki-Plemelj formula and Liouville's theorem that the unique solution of the RH problem for $Y$ is given by
\begin{equation}\label{Y}
Y(z)= \left (\begin{array}{cc}
\pi_n(z)& \frac 1{2\pi i}\int_\mathbb{R} \frac{\pi_n(x) w(x)}{x-z}dx\\
-2\pi i \gamma_{n-1}^2 \;\pi_{n-1}(z)& - \gamma_{n-1}^2\;
\int_\mathbb{R} \frac{\pi_{n-1}(x) w(x)}{x-z}dx\end{array} \right ).
\end{equation}
where  $\pi_n(z)$ is the monic orthogonal polynomial with respect to the weight $w(x)$ defined in \eqref{weight}; see  \cite{Fokas}.
For general $\omega\in \mathbb{C}\setminus(-\infty,0)$, we will prove later that for $n$
 large enough, the RH problem for $Y$  can be solved and thus the
polynomials orthogonal with respect to \eqref{weight} exist for large $n$.

Now we introduce
\begin{equation}\label{pis-piv}
\Phi(z,\mu)= Y(z+\mu)e^{-\left (\frac {z^2}{2}+\mu z\right )\sigma_3}z^{\alpha\sigma_3},
\end{equation}
then $\Phi(z,\mu)$ satisfies  constant jumps on the real axis and has a regular singularity at  $z=0$  and  an irregular singularity at infinity, with
 \begin{equation}\label{pis-piv-infinity}\Phi(z,\mu)=\left (I+\frac {Y_1+n\mu\sigma_3}{z}+O\left ( \frac 1 {z^{2}}\right )\right )e^{-\left (\frac {z^2}{2}+\mu z\right )\sigma_3}z^{ (n+\alpha)\sigma_3},\quad \quad z\rightarrow
                             \infty .\end{equation}
Thus, it is related to the Garnier-Jimbo-Miwa Lax pair $\Phi_{\mathrm{PIV}}$ for the Painlev\'e IV equation with the parameters $\theta_{\infty}=-(\alpha+n)$ and $ \theta_0=\alpha$ by
\begin{equation}
 \Phi(z,\mu)=\Phi_{\mathrm{PIV}}(e^{\frac 12 \pi i}z,e^{\frac 12 \pi i}\mu) e^{-\frac {\pi}2 i(\alpha+n)\sigma_3};
\end{equation}
see \cite [Theorem 5.2,Proposition 5.4]{FokasItsKapaevNovokshenov} and \cite[Eq.(C.30)-Eq.(C.34)]{jmu}.
From this relation, we get the $\sigma$-form of the Painlev\'e IV equation \cite [Eq.(C.30)-Eq.(C.37)]{jmu}.
\begin{pro}\label{PIV} Let
\begin{equation}
 \sigma_{\mathrm{IV}}(\mu)=2 \lim_{z\to\infty}z(Y(z)z^{-n\sigma_3}-I)_{11},
\end{equation}
then $ \sigma_{\mathrm{IV}}$ satisfies the Jimbo-Miwa-Okamoto $\sigma$-form of  the Painlev\'e IV equation
\begin{equation}\label{sigma-painleve IV}
  (\sigma_{\mathrm{IV}}'')^2=4(\mu\sigma_{\mathrm{IV}}'-\sigma_{\mathrm{IV}})^2-4\sigma_{\mathrm{IV}}'(\sigma_{\mathrm{IV}}'-4\alpha)(\sigma_{\mathrm{IV}}'+2n).
\end{equation}
\end{pro}

At the end of this section, we state and prove a differential identity for the  Hankel determinant  associated
with the weight $w(x)$ in \eqref{weight}. The identity serves as the start point of our asymptotic study of the Hankel determinant.
We also show that the logarithmic derivative of the Hankel determinant satisfies the Jimbo-Miwa-Okamoto $\sigma$-form of  the Painlev\'e IV equation, as was first obtained in \cite{ForresterWitte}.

\begin{lem}\label{lem:Hankel-differential identity }
Let  $H_n$ be the Hankel determinant associated with $w(x;\mu)$ in \eqref{pGUE}, then
\begin{equation}\label{differential identity}
    \frac d{d\mu}\ln H_{n}=2 \lim_{z\to\infty}z(Y(z)z^{-n\sigma_3}-I)_{11},
\end{equation}
and  $ \sigma_{\mathrm{IV}}(\mu)= \frac d{d\mu}\ln H_{n}$ satisfies the Jimbo-Miwa-Okamoto $\sigma$-form of  the Painlev\'e IV equation
\begin{equation}\label{sigma-painleve IV-1}
  (\sigma_{\mathrm{IV}}'')^2=4(\mu\sigma_{\mathrm{IV}}'-\sigma_{\mathrm{IV}})^2-4\sigma_{\mathrm{IV}}'(\sigma_{\mathrm{IV}}'-4\alpha)(\sigma_{\mathrm{IV}}'+2n).
\end{equation}
\end{lem}

\begin{proof}
By \eqref{Hankel determiant} and taking logarithmic derivative with respect to $\mu$, we have
 \begin{equation}\label{log derivative H}
 \frac d{d\mu}\ln H_{n}=-2\sum_{k=0}^{n-1}\gamma_{k}^{-1}\frac d {d\mu}\gamma_{k},
\end{equation}
where
\begin{equation}\label{leading coef1}
\gamma_{k}^{-2}=\int_{\mathbb{R}}\pi_k(x)^2w(x)dx=\int_{\mathbb{R}}(\pi_k(x+\mu))^2w(x+\mu)dx.
\end{equation}
Taking derivative with respect to $\mu$ on both sides of \eqref{leading coef1} and using the orthogonality and the definition of the weight \eqref{weight}, we get
\begin{equation}\label{leading coef}\gamma_{k}^{-1}\frac d {d\mu}\gamma_{k}=\gamma_{k}^2\int_{\mathbb{R}}(x+\mu)\pi_k(x+\mu)^2w(x+\mu)dx=\gamma_{k}^2\int_{\mathbb{R}}x\pi_k(x)^2w(x)dx.
\end{equation}
It follows from \eqref{log derivative H}, \eqref{leading coef} and the Christoffel-Darboux formula that
\begin{equation}\label{log derivative H-integral}\frac d{d\mu}\ln H_{n}=-2\gamma_{n-1}^2\int_{\mathbb{R}}x\left(\frac d {dx}\pi_n(x) \pi_{n-1}(x)-\pi_n(x) \frac d {dx}\pi_{n-1}(x)\right)w(x)dx.
\end{equation}
Write $\pi_n(x)=x^n+p_nx^{n-1}+\cdots$,  we have the following decomposition into the orthogonal system $\{\pi_k\}_{k=0}^n$,
\begin{equation}\label{decomposition}x\frac d{dx}\pi_n(x)=n\pi_n-p_n\pi_{n-1}+\cdots \end {equation}
By using the orthogonality, we have $\frac d{d\mu}\ln H_{n}=2p_n$.
Then \eqref{differential identity} follows from \eqref{Y}, and the fulfilment of \eqref{sigma-painleve IV-1}  follows from  Proposition \ref{PIV}.
\end {proof}

\subsection{The   transformations $Y\rightarrow T\to S$}
From this section on, we take $\mu:=\mu_n=\sqrt{2n}+\frac {s}{\sqrt{2}n^{1/6}}$ in the weight function \eqref{weight}.
We introduce the transformation
 \begin{equation}\label{Y to T}T(z)= (2n)^{-\frac 12(n+\alpha)\sigma_3}e^{-\frac 1 2 nl \sigma_3} Y(\sqrt{2n}z) e^ {n \left (\frac 1 2 l
-g(z)\right )\sigma_3}(2n)^{\frac {\alpha}2 \sigma_3},\quad z\in
\mathbb{C}\backslash \mathbb{R}, \end{equation}
where the $g$-function
\begin{equation}\label{g}g(z)=\int_{-1}^{1}  \ln(z-x) \varphi(x) dx, \quad \varphi(x)=\frac 2 {\pi } \sqrt{1-x^2}, \quad z\in \mathbb{C}\setminus (-\infty, 1],  \end{equation}
and $l=\-1-2\ln 2$, with the logarithm taking the branch $\arg (z-x)\in (-\pi , \pi)$.
We also introduce the $\phi$-function
\begin{equation}\label{phi}
\phi(z)=2\int_{1}^z\sqrt{x^2-1}dx=z\sqrt{z^2-1}-\ln\left(z+\sqrt{z^2-1}\right), \quad z\in \mathbb{C}\setminus (-\infty, 1],
\end{equation}where the branches are chosen such that $\sqrt{z^2-1}$ is analytic in $\mathbb{C}\setminus [-1, 1]$, taking positive values on  $(1, +\infty)$, $z+\sqrt{z^2-1}$ maps $\mathbb{C}\setminus [-1, 1]$ onto  the outside of the unit disk, with $z+\sqrt{z^2-1}\sim 2z$ as $z\to\infty$, and the logarithm taking the principal branch.
The $\phi$-function and $g$-function are related by the phase condition
\begin{equation}\label{phase condition}
2\left[g(z)+\phi(z)\right]-2z^2-l=0, \quad z\in \mathbb{C}\backslash(-\infty,1].
\end{equation}

So defined,  $T$ is normalized at infinity in the sense that $T(z)=I+O(1/z)$,
and satisfies the jump
condition
\begin{equation}\label{Jump-T}T_+(x)=T_-(x)\left\{
\begin{array}{ll}
 \left(
                               \begin{array}{cc}
                                 1 & |x-\lambda|^{2\alpha}\theta(x-\lambda)e^{-2n \phi(x)} \\
                                 0 & 1 \\
                               \end{array}
                             \right),&   x\in(1,+\infty),\\
                             &\\
  \left(
                               \begin{array}{cc}
                                 e^{2n\phi_+(x)} & |x-\lambda|^{2\alpha}\theta(x-\lambda) \\
                                 0 & e^{2n \phi_-(x)} \\
                               \end{array}
                             \right),&   x\in(-1, 1), \\
                             &\\
   \left(
                               \begin{array}{cc}
                                 1 & |x-\lambda|^{2\alpha} e^{-2n \phi_+(x)} \\
                                 0 & 1 \\
                               \end{array}
                             \right), &  x\in(-\infty,-1),
\end{array}
\right .
\end{equation}
where $\lambda=\frac {\mu_n}{\sqrt{2n}}=1+\frac {s}{2n^{2/3}}$, $\theta(x)=\omega$ for $x>0$ and $\theta(x)=1$ for $x<0$.

To remove the    oscillation of the diagonal entries of the jump matrix on $(-1, 1)$, we deform the interval $[-1,\lambda]$ into a lens-shaped region as illustrated in Fig.\;\ref{fig2},  and introduce  the following transformation
\begin{equation}\label{T-S}
S(z)=\left \{
\begin{array}{ll}
  T(z), & \mbox{for $z$ outside the lens,}
  \\ &\\
  T(z) \left( \begin{array}{cc}
                                 1 & 0 \\
                                  -(\lambda-z)^{-2\alpha} e^{2n\phi(z)} & 1 \\
                               \end{array}
                             \right) , & \mbox{for $z$ in the upper lens,}\\ &\\
T(z) \left( \begin{array}{cc}
                                 1 & 0 \\
                                  (\lambda-z)^{-2\alpha} e^{2n\phi(z)} & 1 \\
                               \end{array}
                             \right) , & \mbox{for $z$ in the lower lens. }
\end{array}\right .\end{equation}

It is easy to see that $S$ satisfies the jump conditions
\begin{figure}[ht]
\begin{center}
\resizebox*{7cm}{!}{\includegraphics{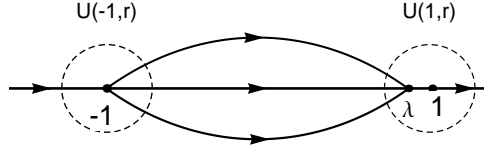}}
\caption{Regions and contours for $S$. $\Sigma_S$ consists of the bold oriented contours: $U(\pm 1. r)$ are disks of radius $r$, centered respectively at $\pm 1$. }
\label{fig2} \end{center}
\end{figure}
\begin{equation}\label{Jump-S}S_+(z)=S_-(z)\left\{
\begin{array}{ll}
 \left(
                               \begin{array}{cc}
                                 1 & \omega|x-\lambda|^{2\alpha}e^{-2n\phi(x)} \\
                                 0 & 1 \\
                               \end{array}
                             \right),&   z=x\in(\max(\lambda,1),+\infty),\\[.4cm]

  \left(
                               \begin{array}{cc}
                                 0 & |x-\lambda|^{2\alpha} \\
                                 -|x-\lambda|^{-2\alpha} & 0 \\
                               \end{array}
                             \right),&   z=x\in(-1, \min(\lambda, 1) ), \\[.4cm]

  \left( \begin{array}{cc}
                                 1 & 0 \\
                                  (\lambda-z)^{-2\alpha} e^{2n\phi(z)} & 1 \\
                               \end{array}
                             \right),&   z~\mbox{on lens}, \\[.4cm]

   \left(
                               \begin{array}{cc}
                                 1 & |x-\lambda|^{2\alpha} e^{-2n \phi_+(x)} \\
                                 0 & 1 \\
                               \end{array}
                             \right), & z= x\in(-\infty,-1),
\end{array}
\right.
\end{equation}
and
\begin{equation}\label{Jump-S-2}S_+(x)=S_-(x)\left\{
\begin{array}{ll}
 \left(
 \begin{array}{cc}
                                e^{2n\phi_+(x)} & \omega|x-\lambda|^{2\alpha} \\
                                 0 & e^{2n\phi_-(x)}\\
                               \end{array}
                             \right),&   x\in(\lambda,1),~\mbox{if}~\lambda<1,\\[.4cm]
  \left(
                               \begin{array}{cc}
                                 0 & |x-\lambda|^{2\alpha} e^{-2n\phi(x)}\\
                                 -|x-\lambda|^{-2\alpha}e^{2n\phi(x)} & 0 \\
                               \end{array}
                             \right),&   x\in(1,\lambda),~\mbox{if}~\lambda>1.
\end{array}
\right.
\end{equation}

We note that
 for $n$ large, the jumps  for $S$ is  close  to the identity matrix,  except the one on   $[-1,\lambda]$. Here both cases, namely  $\lambda<1$ and $\lambda>1$, are possible.

\subsection{Global Parametrix}
The global parametrix solves the following approximating RH problem, with only the jump along $(-1, \lambda)$ remains:
\begin{description}
\item(N1)~~  $N(z)$ is analytic in  $\mathbb{C}\backslash
[-1,\lambda]$;
\item(N2)~~   \begin{equation}\label{Jump-N} N_{+}(x)=N_{-}(x)\left(
       \begin{array}{cc}
       0 & |x-\lambda|^{2\alpha} \\
       -|x-\lambda|^{-2\alpha} & 0 \\
       \end{array}
       \right),\quad x\in  (-1,\lambda);\end{equation}
\item(N3)~~    \begin{equation}\label{NinfiniytOT} N(z)= I+O(z^{-1}) ,\quad z\rightarrow\infty .\end{equation}
  \end{description}

A solution of the RH problem can be constructed explicitly as
\begin{equation}\label{N-expression}
 N(z) =d(\infty)^{\sigma_3} N_0(z) d(z)^{-\sigma_3},
 \end{equation}
 where
   \begin{equation}\label{N-0}
  N_0(z)=
\left(
  \begin{array}{cc}
    \frac {\varsigma(z) + \varsigma^{-1}(z)} {2}&\frac {\varsigma(z) - \varsigma^{-1}(z)} {2i} \\
    -\frac {\varsigma(z) - \varsigma^{-1}(z)} {2i} &\frac {\varsigma(z) + \varsigma^{-1}(z)} {2} \\
  \end{array}
\right), \quad
 \varsigma(z)=\left ( \frac {z-\lambda}{z+1} \right )^{1/ 4},
  \end{equation}
 where $\varsigma(z)$ is analytic in $\mathbb{C}\setminus [-1, \lambda]$, and $\varsigma(z)\sim 1$ as $z\to\infty$.

 The Szeg\H{o} function  $d(z)$ in \eqref{N-expression} is analytic  in $\mathbb{C}\setminus [- 1,\lambda]$ and satisfies
  $$d_+(x)d_-(x)=|x-\lambda|^{2\alpha}, \quad x\in(-1,\lambda);$$ see \cite{Szego1975}.  Indeed, the  Szeg\H{o} function can be written  explicitly as
\begin{equation}\label{szego function}d(z)=(z-\lambda)^\alpha(\sqrt{z-\lambda}+\sqrt{z+1})^{-2\alpha}(1+\lambda)^{\alpha},\quad  d(\infty)=
\left (\frac{1+\lambda}{4}\right )^{\alpha},
\end{equation} where $\arg(z+1)\in (-\pi, \pi)$, $\arg(z-\lambda)\in (-\pi, \pi)$ and the logarithm takes the principal branch with the branch cut
along $(-\infty,0)$.

\subsection{Local parametrix at $z=-1$}
The local parametrix at $z=-1$ satisfies the following RH problem.
\begin{description}
  \item(a)~~ $P^{(-1)}(z)$  is analytic in $U(-1,r) \backslash  \Sigma_{S}$; see Fig.\;\ref{fig2} for the contours;
  \item(b)~~
 $P^{(-1)}_+(z)=P^{(-1)}_-(z)J_{S}(z), ~~z\in\Sigma_S \cap U(-1,r)$, where $J_{S}(z)$ are the jump matrices of $S(z)$ in \eqref{Jump-S};
  \item(c)~~
$P^{(-1)}(z)N^{-1}(z)=I+ O(1/n)$ for $z\in \partial U(-1, r)$.
 \end{description}

\begin{figure}[ht]
\begin{center}
\resizebox*{6cm}{!}{\includegraphics{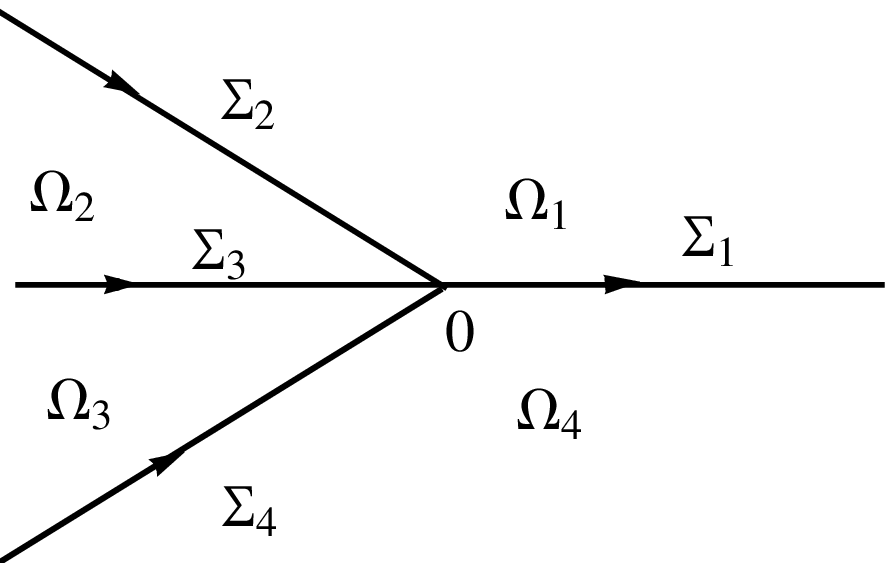}}
\caption{Regions and contours for $\Phi_A$}
\label{fig-Airy} \end{center}
\end{figure}

The local parametrix can be constructed as follows:
\begin{equation}\label{local parametrix at L}
P^{(-1)}(z)=\sigma_3E_{-1}(z)\Phi_A\left (\frac 32n\phi(-z)^{2/3}\right )e^{n\phi(z)\sigma_3}\sigma_3,
 \end{equation}
 where $E_{-1}(z)$ is a certain matrix-valued function  analytic in $U(-1,r)$,
$\Phi_A$ is the Airy parametrix  defined as
\begin{equation}\label{Airy-model-solution}
  \Phi_A(z)=M\left\{
                 \begin{array}{ll}
                  \left(
                     \begin{array}{cc}
                        \Ai(z) & \Ai(e^{-\frac{2\pi i}{3}} z) \\
                   \Ai'(z)&e^{-\frac{2\pi i}{3}} \Ai'(e^{-\frac{2\pi i}{3}} z) \\
                      \end{array}
                 \right)e^{-\frac{\pi i}{6}\sigma_3}, & z\in \Omega_1, \\[.5cm]
                       \left(
                     \begin{array}{cc}
                        \Ai(z) & \Ai(e^{-\frac{2\pi i}{3}} z) \\
                      \Ai'(z)&e^{-\frac{2\pi i}{3}} \Ai'(e^{-\frac{2\pi i}{3}}z) \\
                      \end{array}
                  \right)e^{-\frac{\pi i}{6}\sigma_3}\left(
                                                       \begin{array}{cc}
                                                    1 & 0 \\
                                                         -1 & 1 \\
                                                       \end{array}
                                                      \right)
                  , & z\in \Omega_2, \\[.5cm]
                            \left(
                     \begin{array}{cc}
                         \Ai(z) & -e^{-\frac{2\pi i}{3}}\Ai(e^{\frac{2\pi i}{3}} z) \\
                        \Ai'(z)&- \Ai'(e^{\frac{2\pi i}{3}} z) \\
                     \end{array}
                  \right)e^{-\frac{\pi i}{6}\sigma_3}\left(
                                                       \begin{array}{cc}
                                                        1 & 0 \\
                                                      1 & 1 \\
                                                       \end{array}
                                                    \right)
                  , &z\in \Omega_3, \\[.5cm]
                        \left(
                     \begin{array}{cc}
                    \Ai(z) & -e^{-\frac{2\pi i}{3}} \Ai(e^{\frac{2\pi i}{3}} z) \\
                        \Ai'(z)&- \Ai'(e^{\frac{2\pi i}{3}}z) \\
                      \end{array}
                  \right)e^{-\frac{\pi i}{6}\sigma_3}, & z\in \Omega_4,
                 \end{array}
          \right.
   \end{equation}
   with the regions indicated in Fig. \ref{fig-Airy} and  $M=\sqrt{2\pi}e^{\frac 16\pi i}\left(
                                                               \begin{array}{cc}
                                                                 1 &0\\
                                                                 0 & -i \\
                                                               \end{array}
                                                             \right)
   $; cf.
  \cite[(7.9)]{DeiftZhouS} and \cite[Ch.\;7.6]{d}.
  The Airy parametrix $\Phi_A$ satisfies the asymptotic condition at infinity
  \begin{equation}\label{Airy-model-infinity}
\Phi_A(z)=z^{-\sigma_3/4}\left (I+O(z^{-3/2})\right )\frac {I+i\sigma_1}{\sqrt{2}}e^{-\frac 23z^{\frac 32\sigma_3}}.
 \end{equation}
 To match $P^{-1}(z)$ with the global parametrix $N(z)$, we choose
 $$E_{-1}(z)=\sigma_3N(z)\sigma_3\left (\frac 32n\phi(-z)^{2/3}\right )^{\sigma_3/4}\frac {I-i\sigma_1}{\sqrt{2}}.$$


\subsection{Local parametrix at $z=\lambda$}

We seek a local parametrix $P^{(1)}(z)$  near $ \lambda= \frac {\mu_n}{\sqrt{2n}}=1+\frac {s}{2n^{2/3}}$. The parametrix solves   the following RH problem:
\begin{description}
  \item(a)~~ $P^{(1)}(z)$  is analytic in $U(1,r) \backslash  \Sigma_{S}$;   see Fig.\;\ref{fig2} for the contours;
  \item(b)~~ On $\Sigma_{S}\cap U(1,r)$, $P^{(1)}(z)$ satisfies the same jump condition as $S(z)$,
  \begin{equation}\label{PJS}
P^{(1)}_+(z)=P^{(1)}_-(z)J_{S}, ~~z\in\Sigma_S \cap U(1,r);\end{equation}
  \item(c)~~  $P^{(1)}(z)$ satisfies the following matching condition on $\partial U(1,r)$:
\begin{equation}\label{mathcing condition}
P^{(1)}(z)N^{-1}(z)=I+ O\left (n^{-1/3}\right );
 \end{equation}
 \item(d)~~  The behavior at $z=\lambda$ is the same as that of $S(z)$.
 \end{description}

It is observed that $P^{(1)}(z)e^{-n\phi(z)\sigma_3}(z-\lambda)^{\alpha \sigma_3}$ solves a RH problem
with the same constant jumps as the model RH problem \eqref{Psi-jump}.
We take the conformal mapping
\begin{equation}\label{f-t}
f(z)=\left (\frac{3}{2}\phi(z)\right )^{2/3}=2(z-1)+O\left ((z-1)^2\right )
\end{equation}from a neighborhood of $z=1$ to that of the origin.
Then for $r$ small enough, the local parametrix  ${P}^{(1)}(z)$ can be constructed as follows
\begin{equation}\label{parametrix}
P^{(1)}(z)=E(z)\Psi\left (n^{2/3}(f(z)-f(\lambda));n^{2/3}f(\lambda)\right )e^{n\phi(z)\sigma_3}(z-\lambda)^{-\alpha \sigma_3} ,\end{equation}
where
\begin{equation}\label{E}
E(z)=N(z)(z-\lambda)^{\alpha \sigma_3}\frac{1}{\sqrt{2}}(I-i\sigma_1)
\left [n^{2/3}(f(z)-f(\lambda))\right ]^{\sigma_3/4} .\end{equation}

To verify the matching condition \eqref{mathcing condition}, in view of \eqref{psi-infinity},
we can write the exponent of \eqref{parametrix} as
\begin{align}\label{F-n}
F_n(z):=&n\phi(z)- \vartheta\left ((n^{2/3}(f(z)-f(\lambda)),n^{2/3}f(\lambda)\right ) \nonumber\\
=&\frac n4f(\lambda)^2(f(z)-f(\lambda))^{-\frac 12}\left [1+O(\lambda-1)\right ]\\
=&n^{-\frac 1 3} \frac {s^2}4 \left (f(z)-f(\lambda)\right )^{-\frac 12} +O\left (\frac 1 n\right )\nonumber \end{align}
for $|z-1|=r$ and $n\to \infty$. Now substituting it in \eqref{psi-infinity} and \eqref{parametrix}, we have
 the matching condition \eqref{mathcing condition}.

\subsection{The final transformation: $S\rightarrow R$}
The final transformation is defined as
\begin{equation}\label{S-R}
R(z)=\left\{ \begin{array}{ll}
                S(z)N^{-1}(z), & z\in \mathbb{C}\backslash \left \{ U(-1,r)\cup U(1,r)\cup \Sigma_S \right \},\\
               S(z) \left\{P^{(-1)}(z)\right\}^{-1}, & z\in   U(-1,r)\backslash \Sigma_{S},  \\
               S(z)  \left\{P^{(1)}(z)\right\}^{-1}, & z\in   U(1,r)\backslash
               \Sigma_{S} .
             \end{array}\right .
\end{equation}
From the matching condition \eqref{mathcing condition}, we have
\begin{equation}\|J_R(z)-I\|_{L^2\cap L^{\infty}(\Sigma_R)}=O(n^{-1/3}).
\end{equation}
Thus
 \begin{equation}\label{R-asymptotic}R(z)=I+O(n^{-1/3}),
\end{equation}  where the error bound is uniform for $z$ in whole complex plane.

\section{Proof of the main theorems}\label{sec:4}
The differential identity \eqref{differential identity} provides an expression of the logarithmic derivative of the Hankel determinant,
in terms of  the large-$z$ behavior of $Y(z)$; see  \eqref{Y-infinity}.  From    \eqref{Y}, one can also write the recurrence coefficients in \eqref{recurrence relations} as
\begin{equation}\label{recurrence coefficients and Y}
a_n=(Y_1)_{11}+\frac {(Y_2)_{12}}{(Y_1)_{12}}\quad\mbox{and}\quad
b^2_n=(Y_1)_{12}(Y_1)_{21},
\end{equation}
where $Y_1$ and  $Y_2$ are the coefficients in the large-$z$ expansion of $Y(z)$ in \eqref{Y-infinity}. In this section,
by using the asymptotics of $Y$ obtained in Section \ref{sec:3}, we proceed to derive the large-$n$ asymptotic
behavior of the Hankel determinant, the recurrence coefficients and the correlation kernel.

Tracing back the sequence of transformations $Y\to T\to S\to R$, given in  \eqref{Y to T}, \eqref{T-S}  and \eqref{S-R}, we have
\begin{equation}\label{Y tracing back}
Y(\sqrt{2n}z)=(2n)^{\frac 12(n+\alpha)\sigma_3}e^{\frac 12 nl\sigma_3}R(z)N(z)e^{ng(z)\sigma_3-\frac 12 nl \sigma_3}(2n)^{-\frac {\alpha}2\sigma_3}.
\end{equation}
As $z\rightarrow\infty$,
\begin{equation}\label{g-expand}
e^{ng(z)\sigma_3}z^{-n\sigma_3}=I+\frac{G_1}{z}+\frac{G_2}{z^2}
+O\left (\frac 1 {z^3}\right ),
\end{equation}
where the matrix-valued coefficients $G_1=0$ and $G_2=\left\{-\frac{n}{2}\int_{-1}^1 x^2\varphi(x)dx\right\} \sigma_3$. Denoting
\begin{equation}\label{N-expand}
N(z)=I+\frac{N_1}{z}+\frac {N_2}{z^2}+O\left (\frac 1 {z^3}\right ) ~~\mbox{as}~z\to\infty
\end{equation}
and
\begin{equation}\label{R-expand}
R(z)=I+\frac{R_1}{z}+\frac {R_2}{z^2}+O\left (\frac 1 {z^3}\right )~~\mbox{as}~z\to\infty,
\end{equation}
and substituting \eqref{N-expand} and \eqref{R-expand}
in \eqref{Y tracing back}, we have
\begin{equation}\label{Y1}
Y_1=(2n)^{\frac {\alpha}2\sigma_3}\sqrt{2n}e^{\frac 12 n l \sigma_3}(R_1+N_1)e^{-\frac 12 n l \sigma_3}(2n)^{-\frac {\alpha}2\sigma_3},\end{equation}
and
\begin{equation}\label{Y2}
Y_2=2n(2n)^{\frac {\alpha}2\sigma_3}e^{\frac 12 n l \sigma_3}(R_2+G_2+N_2+R_1N_1)e^{-\frac 12 n l \sigma_3}(2n)^{-\frac {\alpha}2\sigma_3}.\end{equation}

To evaluate  $N_1$ and $N_2$, we expand the Szeg\H{o} function \eqref{szego function} at infinity,
\begin{align}\label{D-expand}
d(z)=&\left(\frac {1+\lambda}{4}\right)^{\alpha}\left \{1-\frac{\alpha(\lambda+1)}{2}\frac 1z\right .\nonumber\\
&\left .+\frac {4\alpha(2\alpha-1)+4\alpha(2\alpha-3)(\lambda-1)+\alpha(2\alpha-5)(\lambda-1)^2}{16z^2}+O\left (\frac 1{z^3}\right )\right \}.
\end{align}
We also expand $N_0(z)$ in \eqref{N-0} at infinity,
\begin{equation}\label{N-0-expand-1}
N_0(z)=I+\frac{N_1^0}{z}+\frac{N_2^0}{z^2}+O\left (\frac 1{z^3}\right ),
\end{equation}
where
\begin{equation}\label{N-0-expand-2}
N_1^0=-\frac{1}{4}(1+\lambda)\sigma_2,\ \ \
N_2^0=\frac{1-\lambda^2}{8}\sigma_2+ \frac{(1+\lambda)^2}{32}I;
\end{equation}
see \cite{XuZhao}.
Substituting \eqref{D-expand}-\eqref{N-0-expand-2} in  \eqref{N-expression} gives  \eqref{N-expand}, with
\begin{equation}\label{N-expand-2}
N_1=\left(\frac {1+\lambda}{4}\right)^{\alpha\sigma_3}\left(-\frac { (1+\lambda)\sigma_2}{4}+\frac { \alpha (1+\lambda)\sigma_3}{2}\right)\left(\frac {1+\lambda}{4}\right)^{-\alpha\sigma_3},
\end{equation}
and
\begin{align}\label{N-expand-3}
N_2=&\left(\frac {1+\lambda}{4}\right)^{\alpha\sigma_3}\left\{-\frac {i \alpha(1+\lambda)^2\sigma_1}{8}+\frac { (1-\lambda^2)\sigma_2}{8}+\frac { (1+\lambda)^2 I}{32}\right . \nonumber\\
&\left . +\frac {2\alpha^2(\lambda+1)^2I+(4\alpha +12\alpha(\lambda-1)+5\alpha(\lambda-1)^2)\sigma_3}{16}\right\}\left(\frac {1+\lambda}{4}\right)^{-\alpha\sigma_3}.
\end{align}
A combination of  \eqref{psi-infinity}, \eqref{parametrix} and \eqref{E} yields   the following more precise description  of the jump for $R(z)$ on $\partial U(1, r)$:
\begin{align}\label{R-jump-1}
J_R(z)=&N(z)(z-\lambda)^{\alpha \sigma_3}\left(I+\frac{\hat{\sigma}(s)(\sigma_3-i\sigma_1)}{2n^{1/3} (f(z)-f(\lambda))^{1/2}}-\frac {(\hat{\sigma}(s)^2+\hat{\sigma}'(s))\sigma_2}{2n^{2/3} (f(z)-f(\lambda))}+O\left (\frac 1 n\right )\right)\nonumber\\
 &\times e^{F_n(z)\sigma_3} (z-\lambda)^{-\alpha \sigma_3} N^{-1}(z),
\end{align}
where $F_n(z)$ is defined in \eqref{F-n}. In  deriving    \eqref{R-jump-1}, use has been made of the fact that $n^{2/3}f(\lambda)=s+O(n^{-2/3})$.
From the approximation of $F_n$ in \eqref{F-n}, we further have
\begin{align}\label{R-jump-2}
J_R(z)= & N(z)(z-\lambda)^{\alpha \sigma_3} \left\{I+\frac{(F_1+\varphi_1)\sigma_3
-i\varphi_1\sigma_1}{n^{1/3}}\right.\nonumber\\
& \left.+\frac{r_4(z,\lambda) I-r_3(z,\lambda) \sigma_2
}{n^{2/3}}+O\left(\frac 1 n\right )\right\}
(z-\lambda)^{-\alpha \sigma_3} N^{-1}(z),
\end{align}
where
\begin{gather}\label{r-3,4}
F_1(z,\lambda,n)=\frac{s^2}{4(f(z)-f(\lambda))^{1/2}},\\
\varphi_1(z,\lambda)=\frac{\hat{\sigma}(s)}{2(f(z)-f(\lambda))^{1/2}},\quad
\varphi_2(z,\lambda)=-\frac{\hat{\sigma}(s)^2+\hat{\sigma}'(s) }{2(f(z)-f(\lambda))},\\
r_3(z,\lambda)=F_1\varphi_1-\varphi_2, \quad\mbox{and}\quad
 r_4(z,\lambda)=F_1\varphi_1+\frac {F_1^2} 2.
\end{gather}

Alternatively,   $N_0(z)$ in \eqref{N-0}  can be written as
\begin{equation*}\label{}
N_0(z)=\frac{I-i\sigma_1}{\sqrt{2}} \varsigma(z)^{-\sigma_3} \frac{I+i\sigma_1}{\sqrt{2}}.
\end{equation*}
Accordingly, for $|z-1|=r$ we have
\begin{equation}\label{rort}
J_R(z) =d(\infty)^{\sigma_3}\frac{I-i\sigma_1}{\sqrt{2}} \left(I+\frac{\Delta_1(z)}{n^{1/3}}+\frac{\Delta_2(z)}{n^{2/3}}+O\left (\frac 1 n\right )\right)\frac{I+i\sigma_1}{\sqrt{2}} d(\infty)^{-\sigma_3},
\end{equation}
where
\begin{align*}
\Delta_1(z)&=\varsigma(z)^{-\sigma_3}\frac{I+i\sigma_1}{\sqrt{2}}d(z)^{-\sigma_3}(z-\lambda)^{\alpha \sigma_3}
 \left [(F_1+\varphi_1)\sigma_3-i\varphi_1\sigma_1\right ]\nonumber\\
  &\quad  \times (z-\lambda)^{-\alpha \sigma_3} d(z)^{\sigma_3}\frac{I-i\sigma_1}{\sqrt{2}}\varsigma(z)^{\sigma_3}\nonumber\\
&=\varsigma(z)^{-\sigma_3}\frac{I+i\sigma_1}{\sqrt{2}}
 \left [(F_1+\varphi_1)\sigma_3-i\varphi_1\eta^{\sigma_3}\sigma_1\eta^{-\sigma_3}\right ]
 \frac{I-i\sigma_1}{\sqrt{2}}\varsigma(z)^{\sigma_3}\nonumber\\
&=(F_1+\varphi_1)\varsigma(z)^{-\sigma_3}\sigma_2\varsigma(z)^{\sigma_3}
  -i\varphi_1\varsigma(z)^{-\sigma_3}\frac{I+i\sigma_1}{\sqrt{2}}\eta^{\sigma_3}\sigma_1\eta^{-\sigma_3}
   \frac{I-i\sigma_1}{\sqrt{2}}\varsigma(z)^{\sigma_3}\nonumber\\
&=r_1(z,t)\sigma_++r_2(z,t)\sigma_--i\varphi_1\left(\frac{i\eta^{-2}-i\eta^2}{2}\sigma_3+
\frac{\eta^2+\eta^{-2}-2}{2}\varsigma(z)^{-\sigma_3}\sigma_1\varsigma(z)^{\sigma_3}\right)
\end{align*}
and
\begin{align*}
\Delta_2(z)&=\varsigma(z)^{-\sigma_3}\frac{I+i\sigma_1}{\sqrt{2}}d(z)^{-\sigma_3}(z-\lambda)^{\alpha \sigma_3}
 \left [r_4(z,t) I-r_3(z,t) \sigma_2\right ]\nonumber\\
 & \times (z-\lambda)^{-\alpha \sigma_3} d(z)^{\sigma_3}\frac{I-i\sigma_1}{\sqrt{2}}\varsigma(z)^{\sigma_3}\nonumber\\
&=\varsigma(z)^{-\sigma_3}\frac{I+i\sigma_1}{\sqrt{2}}
 \left [r_4(z,\lambda) I-r_3(z,\lambda) \eta^{\sigma_3}\sigma_2\eta^{-\sigma_3}\right ]
 \frac{I-i\sigma_1}{\sqrt{2}}\varsigma(z)^{\sigma_3}\nonumber\\
&=r_4(z,\lambda) I-r_3(z,\lambda)\varsigma(z)^{-\sigma_3}\frac{I+i\sigma_1}{\sqrt{2}} \eta^{\sigma_3}\sigma_2\eta^{-\sigma_3}\frac{I-i\sigma_1}{\sqrt{2}} \varsigma(z)^{\sigma_3}
\nonumber\\
&=r_3(z,\lambda) \sigma_3+r_4(z,\lambda)I+\frac{\eta^2+\eta^{-2}-2}{2}r_3(z,\lambda) \sigma_3\\
&\quad +
\frac{i\eta^{2}-i\eta^{-2}}{2}r_3(z,\lambda)\varsigma(z)^{-\sigma_3}\sigma_1\varsigma(z)^{\sigma_3},
\end{align*}
with
\begin{gather}\label{r-1,2,eta}  r_1(z,\lambda)=-\varsigma^{-2}(z)(F_1+2\varphi_1)i,\quad r_2(z,\lambda)=\varsigma^2(z) F_1 i, \quad \eta(z)=\frac {(z-\lambda)^{\alpha }}{d(z)}.
\end{gather}
By \eqref{szego function}, we have
\begin{equation}\label{eta-expand}
   \eta=1+2\alpha \sqrt{\frac {z-\lambda}{1+\lambda}}+2\alpha^2\frac{z-\lambda}{1+\lambda}+O((z-\lambda)^{3/2}),\quad z\rightarrow \lambda,
\end{equation}
and hence
\begin{equation}\label{eta-expand-1}
   \eta^2+\eta^{-2}-2=O(z-\lambda), \quad   \eta^2-\eta^{-2}=8\alpha\sqrt{\frac{z-\lambda}{1+\lambda}}+O(z-\lambda).
\end{equation}
Thus
\begin{equation}
R(z) =d(\infty)^{\sigma_3}\frac{I-i\sigma_1}{\sqrt{2}}
\left(I+\frac{R^{(1)}(z)}{n^{1/3}}+\frac{R^{(2)}(z)}{n^{2/3}}+O\left
(\frac 1 n\right )\right)\frac{I+i\sigma_1}{\sqrt{2}}d(\infty)^{-\sigma_3},
\end{equation}
where  $R^{(1)}$ and   $R^{(2)}$ satisfy the scalar RH problems
\begin{equation}\label{}
R^{(1)}_+(z)-R^{(1)}_-(z)=\Delta_1(z),~~
R^{(2)}_+(z)-R^{(2)}_-(z)=\Delta_2(z)+R^{(1)}_-(z)\Delta_1(z).
\end{equation}
From \eqref{r-3,4} and \eqref{r-1,2,eta}, we know that  $r_1(z,t)$, $r_3(z,t)$ and $ r_4(z,t)$ have a simple pole at $z=\lambda$
and $r_2(z,t)$  is analytic at $z=\lambda$.  By Cauchy's theorem, we get
\begin{equation}\label{R-1}
 R^{(1)}(z)=\left\{\begin{array}{ll}
                   \frac{k_1(s)}{z-\lambda}\sigma_+-\Delta_1(z), &z\in U(1,r),\\[.3cm]
                    \frac{k_1(s)}{z-\lambda}\sigma_+, &z\not\in \overline{U(1,r)},
                  \end{array}
\right.
\end{equation}
and
\begin{equation}\label{R-2}
 R^{(2)}(z)=\frac{k_3(t)
\sigma_3+ k_4(s) I-k_1(s)k_2(s)\sigma_-\sigma_++k_5(s)\sigma_+}{z-\lambda} ,\ \ z\not\in
\overline{U(1,r)}.\end{equation}
where $k_j(s)=\mathrm{Res}_{z=\lambda} r_j(z,\lambda)$ for  $j=1,3,4$, $k_2(s)=\lim_{z\rightarrow\lambda} r_2(z,\lambda)$ and  $k_5(s)$ arises from the  residue of the term $\eta$ in $\Delta_2(z)$ and $R^{(1)}_-(z)\Delta_1(z)$. After some careful calculations, we get from \eqref{r-1,2,eta} and \eqref{r-3,4}
\begin{equation} \label{k-i}
\begin{split}
 k_1(s)&=-i \left (\hat{\sigma}(s)+\frac {s^2}4\right ) +O(n^{-2/3})\\
k_2(s)&=\frac{is^2} {8}+O(n^{-2/3})\\
k_3(s) &=\frac 14 \left (\hat{\sigma}(s)+\frac {s^2} 4\right )\hat{\sigma}(s)+\frac 14\hat{\sigma}'(s)\\
k_4(s) &=\frac{s^2} {16}\left (\hat{\sigma}(s)+ \frac{s^2 } {4}\right )\\
k_5(s)&=4i\alpha k_3(s)+\alpha\hat{\sigma}(s)  k_1(s)+O(n^{-2/3})=i\alpha\hat{\sigma}'(s)+O(n^{-2/3}).
\end{split}\end{equation}
Expanding $R^{(1)}(z)$ and $R^{(2)}(z)$ into Laurent series at infinity
$$R^{(j)}(z)=I+\frac {R_j} z+O\left (\frac 1 {z^2}\right ),\quad j=1,2,$$
then, from \eqref{R-1} and \eqref{R-2} we have
\begin{align}\label{R-i}
R_1=R_2=&d(\infty)^{\sigma_3}\left\{\frac{k_1(s)(i\sigma_3+\sigma_1)}{2n^{1/3}} +\frac{-\left
[2k_3(s)+k_1(s)k_2(s)\right ]\sigma_2+k_5(s)(i\sigma_3+\sigma_1) }{2n^{2/3}}\right.\nonumber\\
&\left .+O\left (\frac {1} {n}\right )  \right\}d(\infty)^{-\sigma_3},\end{align}
where  use has been  made of the estimate $2k_4(s)-k_1(s)k_2(s)=O(n^{-2/3})$, and the identities
\begin{equation*}
 \frac{I-i\sigma_1}{\sqrt{2}}\sigma_+\frac{I+i\sigma_1}{\sqrt{2}}=\frac 12 (i\sigma_3+\sigma_1),~~\frac{I-i\sigma_1}{\sqrt{2}}\sigma_3\frac{I+i\sigma_1}{\sqrt{2}}=-\sigma_2;
\end{equation*}
cf. \eqref{Psi0-2} for $\sigma_+$.

\subsection {Proof of Theorem  \ref{Asymptotic of Hankel determinant}: asymptotics of the Hankel determinant}
Substituting  \eqref{Y1} into  the differential identity \eqref{differential identity}, we have
 \begin{equation}\label{Hankel determinant-R}
\frac {d}{d\mu} \ln H_n(\mu;\alpha,\omega)=2\sqrt{2n}((N_1)_{11}+(R_1)_{11}).
\end{equation}
Replacing  $N_1$ and $R_1$ with \eqref{N-expand-2} and \eqref{R-i}  yields
 \begin{align}\label{Hankel determinant and R}
\frac {d}{d\mu} \ln H_n(\mu;\alpha,\omega)&=2\sqrt{2n}\left (\frac { \alpha (1+\lambda)}{2}+\frac{ik_1(s)}{2n^{1/3}} +\frac{ik_5(s) }{2n^{2/3}}+O\left (\frac 1 n\right )\right )\nonumber\\
&=2\sqrt{2n}\left (\frac { \alpha (1+\lambda)}{2}+\frac{(\hat{\sigma}(s)+\frac {s^2}4)}{2n^{1/3}} -\frac{\alpha\hat{\sigma}'(s) }{2n^{2/3}}+O\left (\frac 1 n\right )\right )\nonumber\\
&=\sqrt{2n}\left (2\alpha +\frac{\sigma(s)}{n^{1/3}} -\frac{\alpha(\sigma'(s)-s) }{n^{2/3}}+O\left (\frac 1 n\right )\right ),
\end{align}
where $n^{2/3}(\lambda-1)=\frac s2+O(n^{-2/3})$,  and $\sigma(s)=\hat{\sigma}(s)+\frac {s^2}4$ is the Jimbo-Miwa-Okamoto $\sigma$-form of the Painlev\'{e} II equation;
see Proposition \ref{Sigma form and psi}. Using the relation $u(s)=-\sigma'(s)$, we obtain \eqref{Hankel-Asy}, as stated in Theorem \ref{Asymptotic of Hankel determinant}.

\subsection {Proof of Theorem \ref{Asymptotic of coefficients}: asymptotics of the recurrence coefficients  }
It follows from \eqref{recurrence coefficients and Y} and \eqref{Y1} that
\begin{equation}\label{b-n and R-i}
b_n^2=2n\left\{(N_1)_{12}(N_1)_{21}+
(N_1)_{12}(R_1)_{21}+(N_1)_{21}(R_1)_{12}+
(R_1)_{12}(R_1)_{21} \right\}.
\end{equation}
From \eqref{N-0-expand-1}, \eqref{R-i} and \eqref{b-n and R-i}, we get
\begin{align}\label{}
b_{n} &  = \sqrt{2n} \displaystyle{\left\{\frac {(1 +\lambda)^2}{16}+\left [\frac { 1 +\lambda }{4}
  ( 2k_3(s)+k_1(s)k_2(s)) +\frac{ k_1 ^2(s)}4\right ]\frac 1 {n^{  2/ 3}}
  +O\left (\frac {1} {n}\right )\right\}^{1/2}}\nonumber\\
   & = \sqrt{2n}\displaystyle{\left\{\frac { 1 +\lambda }{4}+ \left [ k_3(s)+\frac 1 2 k_1(s)k_2(s)
  +\frac 1 4 k_1^2(s)\right ]\frac 1 {n^{  2/ 3}}
  +O\left (\frac {1} {n}\right )\right\}}.
\end{align}
By the definition of $k_j(s)$ in \eqref{k-i}, we obtain the asymptotic approximation of $b_n$,
\begin{align}\label{bn-asymptotic-1}
b_n&=\sqrt{2n}\left(\frac 12+\frac {n^{2/3}(\lambda-1)+\hat{\sigma}'(s)} {4}\frac 1 {n^{  2/ 3}} +O\left (\frac{1}{n}\right ) \right)\nonumber\\
&=\sqrt{2n}\left(\frac 12+\frac {\sigma'(s)} {4}\frac 1 {n^{  2/ 3}} +O\left (\frac{1}{n}\right ) \right)\nonumber\\
&=\sqrt{2n}\left(\frac 12-\frac {u(s)} {4}\frac 1 {n^{2/3}} +O\left (\frac{1}{n}\right )\right),
\end{align}
 where $u(s)=-\sigma'(s)$ satisfies the Painlev\'e XXXIV equation.

In a similar way, we compute the asymptotics of the recurrence coefficient $a_n$. First, a combination of  \eqref{recurrence coefficients and Y}, \eqref{Y1} and \eqref{Y2} gives
\begin{equation}\label{an-R-i}
a_n=\sqrt{2n}\left\{(N_1)_{11}+(R_1)_{11}+\frac{(N_2)_{12}+(R_2)_{12}+(R_1)_{11}(N_1)_{12}+(R_1)_{12}(N_1)_{22}}
{(N_1)_{12}+(R_1)_{12}}\right\}.
\end{equation}
Noting that $(G_2)_{12}=0$, $R_1=R_2=O(n^{-1/3})$, $(N_1)_{11}=-(N_1)_{22}$,  and $ (N_1)_{12}=\frac i4 (1+\lambda)d^2(\infty)=O(1) $, the formula \eqref{an-R-i} is then turned  into \begin{align}\label{an-R-i-1}
\frac{a_n}{\sqrt{2n}}=&\left((N_1)_{11}+\frac{(N_2)_{12}}{(N_1)_{12}}\right)
+\left[\left(2(R_1)_{11}+\frac {(R_1)_{12}}{(N_1)_{12}}\right)-\frac {(R_1)_{12}}{(N_1)_{12}}\left((N_1)_{11}+\frac {(N_2)_{12}}{(N_1)_{12}}\right)\right]\nonumber\\
  & -\frac {(R_1)_{12}}{(N_1)_{12}^2}\left[  (N_1)_{12}\left((R_1)_{11}+\frac {(R_1)_{12}}{(N_1)_{12}}\right)-(R_1)_{12}\left((N_1)_{11}+\frac {(N_2)_{12}}{(N_1)_{12}}\right)\right]+O\left(\frac 1 n \right).
\end{align}
From \eqref{N-expand-2} and \eqref{N-expand-3}, we have
\begin{equation*}
(N_1)_{11}+\frac{(N_2)_{12}}{(N_1)_{12}}=\frac {\lambda-1}{2}=O\left (\frac 1 {n^{2/3}}\right ).
\end{equation*}
By \eqref{N-expand-2} and \eqref{R-i}, we get
\begin{equation*}
2(R_1)_{11}+\frac {(R_1)_{12}}{(N_1)_{12}}=\frac {2k_3(s)+k_1(s)k_2(s)}{n^{2/3}}+O\left (\frac 1 n\right ),
\end{equation*}
and
\begin{equation*}
\frac {(R_1)_{12}(R_1)_{11}}{(N_1)_{12}}=\frac {k_1(s)^2}{2n^{2/3}}+O\left (\frac 1 n\right ).
\end{equation*}
Thus,  \eqref{an-R-i-1} is further simplified to
\begin{align}\label{an-asymptotic-proof}
\frac{a_n}{\sqrt{2n}}&=\frac {\lambda-1}{2}+\frac {k_1(s)^2+4k_3(s)+2k_1(s)k_2(s)}{2n^{2/3}}+O\left (\frac 1 n\right )\nonumber\\
&=\frac {\hat{\sigma}'(s)+\frac s2}{2n^{2/3}}+ O\left (\frac 1 n\right )\nonumber\\
&=\frac {\sigma'(s)}{2n^{2/3}}+O\left (\frac 1 n\right ).
\end{align}
Substituting in  the relation $u(s)=-\sigma'(s)$, we obtain   the asymptotic approximation \eqref{an-asymptotic} of the recurrence coefficient $a_n$,
as stated in Theorem \ref{Asymptotic of coefficients}.

\subsection {Proof of Theorem \ref{Asymptotic of kernel}: asymptotics of the correlation kernel   }

Now Christoffel-Darboux formula applies, and the correlation kernel \eqref{Weighted OP kernel} can be written as
\begin{equation}\label{}K_n(x,y)=\sqrt{w(x)w(y)}\gamma_{n-1}^2\frac {\pi_n(x)\pi_{n-1}(y)-\pi_n(y)\pi_{n-1}(x)}{x-y}.
\end{equation}
From \eqref{Y}, the   kernel can be expressed in terms of $Y(z)$, such that
\begin{equation}\label{kernel and Y}
\begin{split}\sqrt{2n}K_n(\sqrt{2n}x,\sqrt{2n}y)=&(2n)^{\alpha}|x-\lambda|^{\alpha}|y-\lambda|^{\alpha}e^{-n(x^2+y^2)}\sqrt{\theta(x-\lambda)\theta(y-\lambda)}
\\ &\times  \frac {\left \{Y(\sqrt{2n}y)_+^{-1}Y(\sqrt{2n}x)_+\right \}_{21}}{2\pi i(x-y)},
\end{split}
\end{equation}
where $\lambda=\mu/\sqrt{2n}$.
Appealing to  the transformation \eqref{Y to T} and  the phase condition \eqref{phase condition}, we arrive at
\begin{equation}\label{kernel and T}
\begin{split}
\sqrt{2n}K_n(\sqrt{2n}x,\sqrt{2n}y)=&e^{-n(\phi_+(x)+\phi_+(y))}|x-\lambda|^{\alpha}|y-\lambda|^{\alpha}\sqrt{\theta(x-\lambda)\theta(y-\lambda)}\\
&\times \frac {\left\{T(y)_+^{-1}T(x)_+\right\}_{21}}{2\pi i(x-y)}
\end{split}
\end{equation}
for $x,y\in \mathbb{R}$.
Tracing back the transformation \eqref{T-S}, \eqref{S-R} and \eqref{parametrix}, we have
\begin{equation}\label{T in terms of psi-1}|x-\lambda|^{\alpha}e^{-n\phi_+(x)}T(x)\left(
                                                               \begin{array}{c}
                                                                 1 \\
                                                                 0 \\
                                                               \end{array}
                                                             \right)
=R(z)E(z)\Psi_+(n^{2/3}(f(z)-f(\lambda));n^{2/3}f(\lambda))e^{-\alpha\pi i\sigma_3}\left(
                                                               \begin{array}{c}
                                                                 1 \\
                                                                1 \\
                                                               \end{array}
                                                             \right)
\end{equation}
for $x\in(\lambda-\delta,\lambda)$, where $\delta\in (0, r)$ is a constant,
and
\begin{equation}\label{T in terms of psi-2}|x-\lambda|^{\alpha}e^{-n\phi_+(x)}T(x)\left(
                                                               \begin{array}{c}
                                                                 1 \\
                                                                 0 \\
                                                               \end{array}
                                                             \right)
=R(z)E(z)\Psi_+(n^{2/3}(f(z)-f(\lambda));n^{2/3}f(\lambda))\left(
                                                               \begin{array}{c}
                                                                 1 \\
                                                                0 \\
                                                               \end{array}
                                                             \right)
\end{equation}
for $x\in(\lambda,\lambda+\delta)$.
Since $R(z)$ and $E(z)$ are analytic in $(\lambda-\delta,\lambda+\delta)$, one has
\begin{equation}\label{RE derivative} E^{-1}(y)R^{-1}(y)R(x)E(x)-I=O(x-y)\quad\mbox{as}\quad x-y\to 0. \end{equation}
From  \eqref{f-t} we see that
\begin{equation}\label{f x-n}\lim_{n\to \infty}n^{2/3}\left (f\left (\lambda+\frac {v}{2n^{2/3}}\right )-f(\lambda)\right )=v\quad\mbox{and}\quad  \lim_{n\to \infty}n^{2/3}f(\lambda)=s. \end{equation}
Thus,  inserting the expressions \eqref{T in terms of psi-1}-\eqref{T in terms of psi-2} and the estimates \eqref{RE derivative}-\eqref{f x-n} into \eqref{kernel and T}, we arrive at
\begin{equation}\label{kernel-asymptotic-proof}\lim_{n\to\infty}\frac 1{\sqrt{2}n^{1/6}}K_n\left (\mu_n+\frac {v_1}{\sqrt{2}n^{1/6}},\mu_n+\frac {v_2}{\sqrt{2}n^{1/6}}\right )=\frac {\psi_1(v_1;s)\psi_2(v_2;s)-\psi_1(v_2;s)\psi_2(v_1;s)}{2\pi i(v_1-v_2)},
\end{equation}
where
\begin{equation}\label{psi-kernel-1}\left(
                                    \begin{array}{c}
                                      \psi_1(v;s) \\
                                      \psi_2(v;s) \\
                                    \end{array}
                                  \right)
 =\sqrt{\omega}\Psi(v;s) \left(
              \begin{array}{c}
                1 \\
                0 \\
              \end{array}
            \right)\quad \mbox{for } v>0,
 \end{equation}
  and
 \begin{equation}\label{psi-kernel-2}\left(
                                    \begin{array}{c}
                                      \psi_1(v;s) \\
                                      \psi_2(v;s) \\
                                    \end{array}
                                  \right)
 =e^{-\pi i\alpha}\Psi(v;s) \left(
              \begin{array}{c}
                1 \\
                e^{2\pi i\alpha} \\
              \end{array}
            \right)\quad \mbox{for } v<0.
 \end{equation}

\section*{Acknowledgements}
                         The work of Shuai-Xia Xu  was supported in part by the National
Natural Science Foundation of China under grant numbers
11571376 and 11201493, GuangDong Natural Science Foundation under grant numbers 2014A030313176.
                                   Yu-Qiu Zhao  was supported in part by the National
Natural Science Foundation of China under grant number 11571375.

\vskip 2cm

\section*{Appendix. Asymptotics of the $\sigma$-form of Painlev\'e II   and the Painlev\'e XXXIV function }
\begin{appendix}

\setcounter{equation}{0}
\renewcommand{\theequation}{A.\arabic{equation}}
\renewcommand{\thesubsection}{A.\arabic{subsection}}

In the Appendix, we provide the asymptotics of Jimbo-Miwa-Okamoto $\sigma$-form of the Painlev\'e II equation and the Painlev\'e XXXIV function  as $s\rightarrow -\infty$. The proof is similar to that of \cite{ItsKuijlaarsOstensson2009}.

For negative $s$, take  a re-scaling of variable
\begin{equation}\label{A}
 A(z)=(-s)^{\sigma_3/4}\Psi (-sz).
\end{equation}
Then $ A(z)$ is analytic in $\mathbb{C}\backslash{\cup_{j=1}^4\Sigma_j}$ and shares the same jumps and the same  behavior at $0$  with $\Psi(\zeta)$, $j=1,2,3,4$. Here $\Sigma_j$ are the same rays as illustrated in Fig.\;\ref{fig-P34}.  According to the large-$z$ behavior  of $\Psi(z)$ in \eqref{psi-infinity}, one has
\begin{equation}\label{A1-infinity}
A(z)=\left (I+A_1/z+O(z^{-2})\right )z^{-\sigma_3/4}\frac 1{\sqrt{2}}(I+i\sigma_1)e^{-\tau \left(\frac 23 z^{3/2}-z^{1/2}\right)\sigma_3},\quad z\rightarrow \infty,
\end{equation}
where $\tau=(-s)^{3/2}$.
Using  \eqref{psi-infinity} and  \eqref{Sigma form and psi-1}, the $\sigma$-form can be expressed in terms of $A(z)$,
\begin{equation}\label{sigma-A1}
\sigma(s)=i(-s)^{1/2}(A_1)_{12}.
\end{equation}

\subsection{Case 1: $\omega=0 $}
In the case, there is no jump for $\Psi(\zeta)$  along $\zeta\in \Sigma_1$; cf. \eqref{Psi-jump} and Fig.\;\ref{fig-P34}.
Denote
\begin{equation}\label{g1}h(z)=-\frac 23 z^{3/2}+z^{1/2},\end{equation}
then
$$h(z)^2=z(1+O(z)), \quad z\to 0,$$
and  $B(z)=A(z)e^{-\tau h(z)\sigma_3}$  solves the following  RH problem.
\begin{description}
\item(B1)  $B(z)$ is analytic in $\mathbb{C}\backslash \Sigma_B$, $\Sigma_B= \Sigma_2\cup \Sigma_3 \cup \Sigma_4$; cf. Fig.\;\ref{fig-P34} for the contours;
\item(B2) $B(z)$  satisfies the jump condition
         \begin{equation*}
            B_+(z)=B_-(z) \left\{
            \begin{array}{ll}
              \left(\begin{array}{cc} 1 & 0\\ e^{2\pi i\alpha-2\tau h(z)} & 1 \end{array}\right), &  z\in \Sigma_2, \\[.4cm]
              \left(\begin{array}{cc} 0 & 1\\ -1 & 0 \end{array}\right), &  z\in \Sigma_3= (-\infty,0), \\[.4cm]
              \left(\begin{array}{cc} 1 & 0\\ e^{-2\pi i\alpha-2\tau h(z)} & 1 \end{array}\right), & z\in \Sigma_4;
            \end{array}
           \right .
         \end{equation*}
\item(B3)
$ B(z)=\left (I+O(z^{-1})\right )
  z^{-\sigma_3/4}\frac 1{\sqrt{2}}(I+i\sigma_1)$ as  $z\rightarrow \infty$;
\item(B4)  $B(z)$ has the same behavior at $0$ as $A(z)$ does.
\end{description}
The jump matrices are close to the identity matrix except the one on $(-\infty,0]$. Thus we arrive at the approximating  RH problem.
\begin{description}
  \item(a)~~  $B^{(\infty)}(z)$ is analytic in
  $\mathbb{C}\backslash (-\infty,0]$;

  \item(b)~~  $B^{(\infty)}(z)$  satisfies the jump condition
 \begin{equation*}
B^{(\infty)}_+(z)=B^{(\infty)}_-(z)
\left(\begin{array}{cc} 0 & 1\\ -1 & 0 \end{array}\right), \quad z\in (-\infty,0);
\end{equation*}
\item(c)~~
  $B^{(\infty)}(z)=
    \left (I+O\left ( 1/ z\right )\right)
  z^{-\sigma_3/4}\frac 1{\sqrt{2}}(I+i\sigma_1)$ as $z\rightarrow \infty$.
\end{description}

The solution is straightforward. Explicitly we can write down
\begin{equation}\label{B-infiniy}
B^{(\infty)}(z)=z^{-\sigma_3/4}\frac 1{\sqrt{2}}(I+i\sigma_1).
\end{equation}

It is observed that $B^{(\infty)}(z)$ fails to approximate $B(z)$ at the origin since they demonstrate different singularities there. Hence a parametrix has to be brought in.
The local parametrix at the origin  is constructed in terms of the  Bessel functions as
\begin{equation}\label{P-0}
P^{(0)}(z)=E(z)\Phi(\tau^2g^2(z))e^{-\tau h(z)\sigma_3},
\end{equation}
where
\begin{equation*}\label{Phi}
\Phi(z)=\left\{\begin{array}{lll}
\left(\begin{array}{cc} I_{2\alpha} (z^{1/2}) & \frac i \pi K_{2\alpha}(z^{1/2})\\ \pi i z^{\frac 12} I'_{2\alpha}(z^{1/2}) & -z^{\frac 12}K'_{2\alpha}(z^{1/2})\end{array}\right), & |\arg z|<\frac {2\pi}3,\\[.4cm]
\left(\begin{array}{cc} I_{2\alpha} (z^{1/2}) & \frac i \pi K_{2\alpha}(z^{1/2})\\ \pi i z^{1/2} I'_{2\alpha}(z^{1/2}) & -z^{1/2}K'_{2\alpha}(z^{1/2})\end{array}\right)\left(\begin{array}{cc} 1 & 0\\ -e^{2\alpha \pi i} & 1\end{array}\right), & \frac {2\pi}3<\arg z<\pi,\\[.4cm]
\left(\begin{array}{cc} I_{2\alpha} (z^{1/2}) & \frac i \pi K_{2\alpha}(z^{1/2})\\ \pi i z^{\frac 12} I'_{2\alpha}(z^{1/2}) & -z^{\frac 12}K'_{2\alpha}(z^{1/2})\end{array}\right)\left(\begin{array}{cc} 1 & 0\\ e^{-2\alpha \pi i} & 1\end{array}\right), & -\pi<\arg z<-\frac {2\pi }3;\end{array}\right.
\end{equation*}
see \cite{KuijlaarsMcLaughlinAsscheVanlessen}. To match the local parametrix $P^{(0)}(z)$ with the outside parametrix $B^{(\infty)}$, we choose
$$E(z)=\left( \frac {\pi^2\tau^2g^2(z)}{z}\right)^{\sigma_3/4}=\left( \pi\tau\left (1-\frac 23z\right )\right)^{\sigma_3/2},$$
which is analytic for $|z|<\frac 32$. Then from \eqref{B-infiniy} and \eqref{P-0} we have
\begin{equation}\label{matching-2}
P^{(0)}(z)(B^{(\infty)})^{-1}(z)=I+\left(\begin{array}{cc} 0 & i \frac {16\alpha^2-1} {8 h(z)\sqrt{z}} \\  -i \frac {(16\alpha^2+3)\sqrt{z}} {8  h(z)}  & 0 \end{array}\right)\tau^{-1}+O(\tau^{-2}), ~~z\in \partial U_0,
\end{equation}where $U_0$ is a disk centered at the origin with radius less than $3/2$, and $\tau=(-s)^{3/2}\to +\infty $ as $s\to-\infty$.

In the final transformation, we define
\begin{equation}\label{}
C(z)=\left\{\begin{array}{lll}
B(z)(B^{(\infty)}(z))^{-1}, & z\in \mathbb{C}\backslash U_0,\\
B(z)(P^{(0)}(z))^{-1}, & z\in U_0.\end{array}\right.
\end{equation}
Then $C(z)$ is analytic in $\mathbb{C}\backslash \Sigma_C$ and  $J_C(z)=I+O(\tau^{-1})$. Here $\Sigma_C$ is the union of $\partial U_0$ and the parts of $\Sigma_2$ and $\Sigma_4$ outside of $U_0$.   Thus
$C(z)=I+O(1/\tau)$,  uniformly in the complex plane. Take the expansion as $z\to \infty$ $$C(z)=I+C_1/z+O(z^{-2}).$$
Noting that  \eqref{matching-2} actually holds in each small  annulus  around the origin, we have
$$(C_1)_{12}=i\mathrm{Res}_{z=0} \frac {16\alpha^2-1} {8\tau h(z)\sqrt{z}}+O(\tau^{-2})=i \frac {16\alpha^2-1} {8(-s)^{3/2}}+O(s^{-3}).$$
Substituting it  into \eqref{sigma-A1}, we  obtain   the asymptotic approximation  of $\sigma(s)$
\begin{equation}\label{}
\sigma(s)=\frac {s^2}{4}+\frac {16\alpha^2-1} {8}s^{-1}+O(s^{-5/2}),\quad s\to -\infty.
\end{equation}


\subsection{Case 2: $\omega=e^{-2\beta\pi i}$ with $|\Re \beta |<1/2 $}
In \cite{ItsKuijlaarsOstensson2009}, Its, Kuijlaars and \"{O}stenssonthe obtained the  leading asymptotics of the family of Painlev\'{e} XXXIV functions corresponding to $\omega=e^{-2\beta\pi i}$, $|\Re \beta |<1/2 $; see \cite[A.41] {ItsKuijlaarsOstensson2009}. The $\sigma$-form and the Painlev\'{e} XXXIV function are related by $u(x)=-\sigma'(x)$. Unfortunately, to find the $O(1/s) $ term in the asymptotics of  the $\sigma$-form by using the above relation,  we   need  the  $O(1/s^2) $ term  asymptotic approximation  of $u(s)$ which is not available in \cite[A.41] {ItsKuijlaarsOstensson2009}.
In this section, we compute the  asymptotics of the $\sigma$-form directly using its relation with $\Psi$ established in Proposition \ref{Sigma form and psi}.

Instead of $h(z)$ in \eqref{g1}, we introduce in this case   the  function
\begin{equation}\label{g2}
h(z)=\frac 23 (z-1)^{3/2},\quad \arg(z-1)\in(-\pi,\pi).
\end{equation}
It is readily seen that
$$h(z)-(2z^{3/2}/3-z^{1/2})=1/(4\sqrt{z})+O(z^{-3/2}),\quad z\rightarrow \infty.$$
This time, we define
\begin{equation}\label{B1}
B(z)=\left(\begin{array}{cc}
1 & 0\\
-\frac i 4 \tau & 1
\end{array}\right)A(z)e^{\tau h(z)\sigma_3}.
\end{equation}
Then, $B(z)$  satisfies the RH problem.
\begin{description}
\item(B1)  $B(z)$ is analytic in $\mathbb{C}\backslash \Sigma_B$, where $\Sigma_B=\cup^4_{j=2}\Sigma_j \cup (0,1)\cup(1, +\infty)$; cf. Fig.\;\ref{fig-P34} and  Fig.\;\ref{fig-C};

\item (B2) $B(z)$  satisfies the jump condition
\begin{equation*}
   B_+(z)=B_-(z)
   \left\{\begin{array}{ll}
      \left(\begin{array}{cc} e^{-\tau (g_-(z)-g_+(z))} & e^{-2\pi i\beta}\\ 0 & e^{\tau (g_-(z)-g_+(z))} \end{array}\right), & z\in (0,1),\\[.4cm]
  \left(\begin{array}{cc} 1 & e^{-2\pi i\beta-2\tau h(z)}\\ 0 & 1 \end{array}\right), & z\in (1,\infty),\\[.4cm]
 \left(\begin{array}{cc} 1 & 0\\ e^{2\pi i\alpha+2\tau h(z)} & 1 \end{array}\right), & z\in \Sigma_2,\\[.4cm]
 \left(\begin{array}{cc} 0 & 1\\ -1 & 0 \end{array}\right), & z\in\Sigma_3,\\[.4cm]
 \left(\begin{array}{cc} 1 & 0\\ e^{-2\pi i\alpha+2\tau h(z)} & 1 \end{array}\right), &z\in \Sigma_4;
          \end{array}
   \right .
\end{equation*}

\item(B3) $
 B(z)=\left (I+O(z^{-1})\right )
  z^{-\sigma_3/4}\frac 1{\sqrt{2}}(I+i\sigma_1)$ as $z\rightarrow \infty$;
\item(B4)  $B(z)$ has the same behavior at $0$ as $A(z)$ does.
\end{description}

We take the transformation
\begin{equation}\label{C}
C(z)=B(z)\left\{\begin{array}{ll}
I, & z\in \mathbb{C}\backslash \{\Omega_U\cup\Omega_L\},\\ [.2cm]
\left(\begin{array}{cc}
1 & 0\\
- e^{2\pi i\beta} e^{2\tau h(z)} &1
\end{array}\right), & z\in\Omega_U,\\[.4cm]
\left(\begin{array}{cc}
1 & 0\\
e^{2\pi i\beta} e^{2\tau h(z)} &1
\end{array}\right), & z\in\Omega_L,\end{array}\right.
\end{equation}where the regions are illustrated in Fig.\;\ref{fig-C}.
\begin{figure}[t]
\begin{center}
\resizebox*{6cm}{!}{\includegraphics{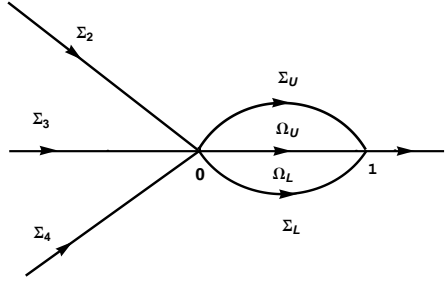}}
\caption{Regions and contours for $C(z)$ when $\omega=e^{-2\beta\pi i}$ with $|\Re \beta |<1/2 $ }
\label{fig-C} \end{center}
\end{figure}
Then  $C(z)$  satisfies the following RH problem.

\begin{description}
\item(C1)  $C(z)$ is analytic in $\mathbb{C}\backslash \Sigma_C$, with $\Sigma_C$ consisting  of all contours illustrated in Fig.\;\ref{fig-C};

\item (C2) $C(z)$ satisfies the jump conditions
\begin{equation}\label{C-jump-remains}
 C_+(z)=C_-(z)
 \left\{\begin{array}{ll}
           \left(\begin{array}{cc} 0 & 1\\ -1 & 0 \end{array}\right), &z\in \Sigma_3,\\[.4cm]
           \left(\begin{array}{cc}0 & e^{-2\pi i\beta}\\ e^{2\pi i\beta} &0 \end{array}\right), & z\in (0,1),
        \end{array}
  \right .
\end{equation}
and
\begin{equation}\label{C-jump-small}
C_+(z)=C_-(z)\left\{\begin{array}{ll}
\left(\begin{array}{cc} 1 & e^{-2\pi i\beta-2\tau h(z)}\\ 0 & 1 \end{array}\right),  & z\in (1,\infty);\\[.4cm]
 \left(\begin{array}{cc} 1 & 0\\ e^{2\pi i\beta+2\tau h(z)} & 1 \end{array}\right), &  z\in \Sigma_{U}\cup \Sigma_{L};\\[.4cm]
 \left(\begin{array}{cc} 1 & 0\\ e^{2\pi i\alpha+2\tau h(z)} & 1 \end{array}\right), &  z\in \Sigma_2;\\[.4cm]
 \left(\begin{array}{cc} 1 & 0\\ e^{-2\pi i\alpha+2\tau h(z)} & 1 \end{array}\right), & z\in \Sigma_4;
\end{array}
  \right .
\end{equation}
\item(C3) $
 C(z)=\left (I+O(z^{-1})\right )
  z^{-\sigma_3/4}\frac 1{\sqrt{2}}(I+i\sigma_1)$  as   $ z\rightarrow \infty$;
\item(C4)  $C(z)$ has the same behavior at $0$ as $A(z)$ does.
\end{description}

The jumps \eqref{C-jump-small} of $C(z)$ are close to the identity matrix for large $\tau$, while those in \eqref{C-jump-remains} are not.
Therefore, we have an approximating RH problem.
\begin{description}
  \item(a)~~  $C^{(\infty)}(z)$ is analytic in
  $C\backslash (-\infty,1]$;

  \item(b)~~  $C^{(\infty)}(z)$  satisfies jump conditions
 \begin{equation*}
C^{(\infty)}_+(z)=C^{(\infty)}_-(z)\left\{\begin{array}{lll}
\left(\begin{array}{cc} 0 & 1\\ -1 & 0 \end{array}\right), & z\in (-\infty,0),\\[.4cm]
\left(\begin{array}{cc} 0 & e^{-2\pi i\beta}\\ -e^{2\pi i\beta} & 0 \end{array}\right), & z\in (0,1);
\end{array}\right.
\end{equation*}
\item(c)~~
  $C^{(\infty)}(z)=
    \left (I+C^{(\infty)}_1/ z+O(z^{-2})\right)
  z^{-\sigma_3/4}\frac 1{\sqrt{2}}(I+i\sigma_1)$ as $z\rightarrow \infty$.
\end{description}

The solution of the above RH problem  is given by
\begin{equation}\label{C-infiniy}
C^{(\infty)}(z)=\left(\begin{array}{cc} 1 & 0\\ 2\beta & 1\end{array}\right)(z-1)^{-\sigma_3/4}\frac 1{\sqrt{2}}(I+i\sigma_1)\left(\frac {\sqrt{z-1}+i}{\sqrt{z-1}-i}\right)^{\beta \sigma_3},
\end{equation}
where the  square  root and the logarithm  take the principal branches with their branch cuts along $(-\infty,1)$ and  $(-\infty,0)$, respectively.

We note that $C^{(\infty)}(z)$ fails to approximate $C(z)$ near $z=0$ and $z=1$. Thus we need to construct parametrices at these two points. First,
we seek a local parametrix defined in  $U(1,r)$ and satisfied the following RH problem.
\begin{description}
  \item(a)~~ $C^{(1)}(z)$  is analytic in $U(1,r) \backslash  \Sigma_{C}$; see Fig.\;\ref{fig-C} for $\Sigma_{C}$;
  \item(b)~~$C^{(1)}(z)$ satisfies the jump condition
 $P^{(1)}_+(z)=P^{(1)}_-(z)J_{C}(z)$ for $z\in\Sigma_C \cap U(1,r)$, where $J_{C}(z)$ are the jump matrices for $C(z)$,  given in \eqref{C-jump-remains} and \eqref{C-jump-small};
  \item(c)~~The matching condition
$C^{(1)}(z)=(I+ O(1/\tau))C^{(\infty)}(z)$ holds on $\partial U(1, r)$.
 \end{description}

 The local parametrix can be constructed in terms of the  Airy function
$$C^{(1)}(z)=E_1(z)\Phi_A\left (\tau^{2/3}(z-1)\right )e^{\pi i\beta\sigma_3}e^{\tau h(z)\sigma_3},$$
where $\Phi_A(\zeta)$ and $h(z)$ are defined in \eqref{Airy-model-solution} and \eqref{g2}, respectively.  To match $C^{(1)}(z)$  with the outside parametrix $C^{(\infty)}(z)$, we choose the analytic factor
$$E_1(z)=C^{(\infty)}(z)e^{-\pi i\beta\sigma_3}\left (\tau^{2/3}(z-1)\right )^{\frac 14\sigma_3}\frac {I-i\sigma_1}{\sqrt{2}}.$$

Next, we turn to  $z=0$, we seek a  local parametrix at the origin of the   form
\begin{equation}\label{C-0}
C^{(0)}(z)=E(z) \Phi(\tau f(z))e^{\tau h(z)\sigma_3}e^{\beta \pi i\sigma_3/2},
\end{equation}
where $\tau=(-s)^{3/2}$,
$$f(z)=\frac 23\left (1-i(z-1)^{\frac {3}{2}}\right )=z(1+O(z)), \quad \arg (z-1)\in(0,2\pi),$$
and $E(z)$ is analytic in a neighborhood $U(0,r)$. The function $\Phi$ in \eqref{C-0} is to be constructed by using arguments from \cite[Sec.\;4]{DeiftItsKrasovky}. To begin with, let
\begin{equation}\label{}
\Phi^{(\mathrm{CHF})}(\zeta)=\Phi(\zeta)\left\{\begin{array}{lll}
e^{-\alpha \pi i\sigma_3},& \Re \zeta>0,~\Im \zeta>0,\\
e^{\alpha \pi i\sigma_3},& \Re \zeta>0,~\Im \zeta<0,\\
I, & \Re \zeta<0.
\end{array}\right.
\end{equation}
Then $\Phi^{(\mathrm{CHF})}(\zeta)$ satisfies the following jump condition on contours illustrated in Fig.\;\ref{fig-CHF},
\begin{equation}\label{CHF-jump}
\Phi^{(\mathrm{CHF})}_+(\zeta)=\Phi^{(\mathrm{CHF})}_-(\zeta)\left\{\begin{array}{ll}
\left(\begin{array}{cc} 0 & e^{-\pi i\beta}\\ -e^{\pi i\beta}  & 0\end{array}\right),  & z\in\Sigma_1,\\[.4cm]
\left(\begin{array}{cc} 1 & 0\\ e^{-\pi i(2\alpha-\beta)} & 1 \end{array}\right),  & z\in\Sigma_2,\\[.4cm]
e^{\pi i\alpha\sigma_3}, & z\in\Sigma_3\cup \Sigma_7,\\[.2cm]
 \left(\begin{array}{cc} 1 & 0\\ e^{\pi i(2\alpha-\beta)} & 1 \end{array}\right), &  z\in \Sigma_{4},\\[.4cm]
 \left(\begin{array}{cc} 0 & e^{\pi i\beta}\\ -e^{-\pi i\beta} & 0 \end{array}\right), &  z\in \Sigma_5,\\[.4cm]
 \left(\begin{array}{cc} 1 & 0\\ e^{-\pi i(2\alpha+\beta)} & 1 \end{array}\right), & z\in \Sigma_6,\\[.4cm]
 \left(\begin{array}{cc} 1 & 0\\ e^{\pi i(2\alpha+\beta)} & 1 \end{array}\right), & z\in \Sigma_8.
\end{array}
  \right .
\end{equation}

\begin{figure}[ht]
\begin{center}
\resizebox*{7cm}{!}{\includegraphics{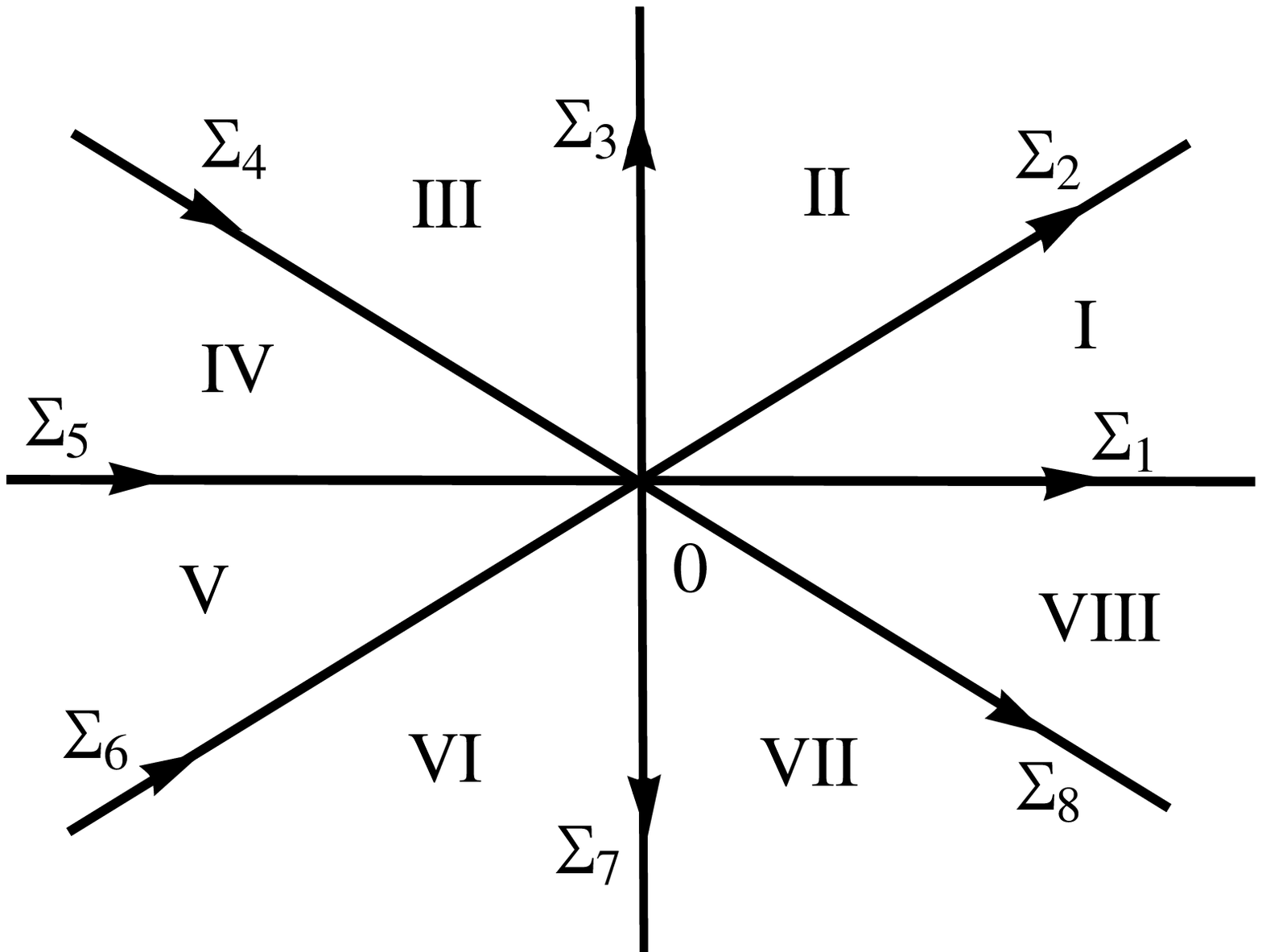}}
\caption{Regions and contours for $\Phi^{(\mathrm{CHF})}$}
\label{fig-CHF} \end{center}
\end{figure}

Such a RH problem with constant jumps  has been treated in \cite{ItsKrasovsky} and \cite{DeiftItsKrasovky}, with the solution being   constructed explicitly for $\zeta\in I$ (cf. Fig.\;\ref{fig-CHF}) as
\begin{align*}\label{}
\Phi^{(\mathrm{CHF})}(\zeta)&=C_1\left(\begin{array}{ll}
(2e^{\pi i/2}\zeta)^\alpha\psi(\alpha+\beta,1+2\alpha,2e^{\pi i/2}\zeta)e^{i\pi(\alpha+2\beta)}e^{-i\zeta}\\
-\frac{\Gamma(1+\alpha+\beta)}{\Gamma(\alpha-\beta)}(2e^{\pi i/2}\zeta)^{-\alpha}\psi(1-\alpha+\beta,1-2\alpha,2e^{\pi i/2}\zeta)e^{i\pi(-3\alpha+\beta)}e^{-i\zeta}\end{array}\right.\nonumber\\
&\quad\quad\quad\quad\quad\quad\left.\begin{array}{rr}
-\frac{\Gamma(1+\alpha-\beta)}{\Gamma(\alpha+\beta)}(2e^{\pi i/2}\zeta)^\alpha\psi(1+\alpha-\beta,1+2\alpha,2e^{-\pi i/2}\zeta)e^{i\pi(\alpha+\beta)}e^{i\zeta}\\
(2e^{\pi i/2}\zeta)^{-\alpha}\psi(-\alpha-\beta,1-2\alpha,2e^{-\pi i/2}\zeta)e^{-i\pi\alpha}\end{array}\right),
\end{align*}
where the constant matrix  $ C_1=2^{\beta \sigma_3}e^{\beta\pi i\sigma_3/2}
\left(\begin{array}{cc}e^{-i\pi(\alpha+2\beta)} & 0\\ 0 & e^{i\pi(2\alpha+\beta)}\end{array}\right)$ and  $\psi(a,b,z)$ is the confluent hypergeometric function; see \cite[Proposition 4.1]{DeiftItsKrasovky}. The solution in the other sectors are determined by using the jump relation in \eqref{CHF-jump}.
Then the function  $\Phi^{(\mathrm{CHF})}(\zeta)$ possesses  the following asymptotic behavior as $\zeta\to\infty$ (see \cite[(4.42)]{DeiftItsKrasovky}):
\begin{equation*}
\begin{split}
  \Phi^{(\mathrm{CHF})}(\zeta)=  &  \left [I+\frac {1}{\zeta} \left(
                                                        \begin{array}{cc}
                                                         - \frac {i(\alpha^2-\beta^2)}{2} & -i2^{2\beta-1}\frac {\Gamma(1+\alpha-\beta)}{\Gamma(\alpha+\beta)}e^{i\pi(\alpha-\beta)}  \\
                                                           i2^{-2\beta-1}\frac {\Gamma(1+\alpha-\beta)}{\Gamma(\alpha-\beta)}e^{-i\pi(\alpha-\beta)} &  \frac {i(\alpha^2-\beta^2)}{2} \\
                                                        \end{array}
                                                      \right)
+O\left (\frac 1 {\zeta^2}\right )\right ] \\
    & \times \zeta^{-\beta \sigma_3}e^{-i\zeta\sigma_3},
\end{split}
\end{equation*}
where $\arg \zeta\in(-\pi/{2},~{3}\pi/{2})$.

Now we take the function $E(z)$ in \eqref{C-0} as
\begin{equation*}\label{}
E(z)=\left\{\begin{array}{ll}
C^{(\infty)}(z)e^{-\left (\alpha+\frac {\beta}2\right )\pi i\sigma_3}e^{\frac {2}{3}i\tau \sigma_3}(\tau f(z))^{\beta \sigma_3},& \Im z>0,\\
C^{(\infty)}(z)\left(\begin{array}{cc}0 & 1\\-1 & 0\end{array}\right)
e^{-\left (\alpha+\frac {\beta}2\right )\pi i\sigma_3}e^{\frac {2}{3}i\tau \sigma_3}(\tau f(z))^{\beta \sigma_3}, & \Im z<0,  \end{array}\right.
\end{equation*}
with $\arg f(z)\in (0,2\pi)$. Then it is readily verified that $E(z)$ is analytic in $U(0,r)$ for sufficiently small $r$, and  the following  matching condition is fulfilled on $\partial U(0,r)$:
\begin{align}\label{matching-A2}
&C^{(0)}(z)\left (C^{(\infty)}(z)\right )^{-1}=I+ C^{(\infty)}(z)f(z)^{\beta\sigma_3}e^{2i\tau \sigma_3/3}\nonumber\\
 &\left [\frac 1 {\tau f(z)}\left(\begin{array}{cc}
-\frac{i(\alpha^2-\beta^2)}{2} &-i\frac{\Gamma(1+\alpha-\beta)}{2\Gamma(\alpha+\beta)}e^{-(\alpha+2\beta)\pi i }(2\tau )^{2\beta}\\
i\frac{\Gamma(1+\alpha+\beta)}{2\Gamma(\alpha-\beta)}e^{(\alpha+2\beta)\pi i}(2\tau)^{-2\beta} & \frac{i(\alpha^2-\beta^2)}{2}\end{array}\right)+O\left (\tau^{2|\Re \beta|-2}\right )\right ]\nonumber\\
&\times e^{-2i\tau \sigma_3/3}f(z)^{-\beta\sigma_3}\left (C^{(\infty)}(z) \right )^{-1}.
\end{align}

In the final transformation, we put
\begin{equation}\label{}
D(z)=\left\{\begin{array}{lll}
C(z)C^{(\infty)}(z)^{-1}, & \mathbb{C}\backslash \overline{U(0,r)\cup U(1,r)} ,\\
C(z)(C^{(0)}(z))^{-1}, & z\in U(0,r),\\
C(z)(C^{(1)}(z))^{-1}, & z\in U(1,r).\end{array}\right.
\end{equation}
In this case, on the remaining  contour consisting of $\partial U(0,r)$, $\partial U(1,r)$, and the parts of $\Sigma_C\setminus(-\infty, 1]$ (see Fig.\;\ref{fig-C}) outside of these two circles,
the jump $J_D(z)=D_-(z)^{-1}D_+(z)=I+O\left (\tau^{2|\Re \beta|-1}\right )$. Thus $D(z)=I+O\left (\tau^{2|\Re \beta|-1}\right )$  uniformly for $z$ in the complex plane.
Expanding  $D(z)$ as
 $$D(z)=I+D_1/z+O(1/z^{2})~~\mbox{for}~z\to \infty,$$
 we get from \eqref{matching-A2} that
 \begin{equation}\label{D-1}
 \begin{split}
 (D_1)_{12}= &  \frac{ 1} {4(-s)^{3/2}}\left(2i(\alpha^2-\beta^2)-\frac{\Gamma(1+\alpha-\beta)}{\Gamma(\alpha+\beta)}
e^{i\theta(s;\alpha,\beta)}+\frac{\Gamma(1+\alpha+\beta)}{\Gamma(\alpha-\beta)}
e^{-i\theta(s;\alpha,\beta)}\right) \\
     & +O\left(s^{3(|\Re \beta|-1)}\right ),
 \end{split}
\end{equation}
with $\theta(s;\alpha,\beta)=\frac 43 |s|^{3/2}-\alpha \pi-6i \beta \ln 2-3i\beta\ln|s|$ for $ \omega=e^{-2\pi i\beta}$ with $|\Re \beta|<1/2$.
Tracing  back the transformations $A(z)\to B(z)\to C(z)\to D(z)$, we have
\begin{equation}\label{}
A(z)=\left(\begin{array}{cc} 1 & 0\\ \frac i 4 \tau & 1 \end{array}\right) D(z) C^{(\infty)}(z) e^{-\tau h(z) \sigma_3}.
\end{equation}
Consequently, in view of \eqref{A1-infinity}, one has
\begin{equation}\label{A-1-asy}
(A_1)_{12}=(D_1+C^{(\infty)}_1)_{12}+\frac i 4 \tau=(D_1)_{12}+2\beta+\frac i 4 (-s)^{3/2}.
\end{equation}
Thus, a combination of \eqref{psi-infinity}, \eqref{A},  \eqref{D-1} and \eqref{A-1-asy} gives  the asymptotic approximation  of $\sigma(s)$ in \eqref{thm: sigma asy negative infinity} as $s\to -\infty$. The asymptotics of $u(s)$ in  \eqref{thm: painleve  negative infinity} as $s\to -\infty$ follows from the relation $u(s)=-\sigma'(s)$.

\end{appendix}

\end{document}